\documentclass[11pt]{amsart}
\usepackage[margin=3cm]{geometry}
\usepackage{enumerate}
\usepackage{amsmath}
\usepackage{amssymb,latexsym}
\usepackage{amsthm}
\usepackage{color}
\usepackage[dvipsnames]{xcolor}
\usepackage{cancel}
\usepackage{graphicx}
\usepackage{cite}
\usepackage{changes}

\makeatletter
\renewcommand{\tocsection}[3]{%
  \indentlabel{\@ifnotempty{#2}{\bfseries\ignorespaces#1 #2\quad}}\bfseries#3}
\renewcommand{\tocsubsection}[3]{%
  \indentlabel{\@ifnotempty{#2}{\ignorespaces#1 #2\quad}}#3}

\newcommand\@dotsep{4.5}
\def\@tocline#1#2#3#4#5#6#7{\relax
  \ifnum #1>\c@tocdepth 
  \else
    \par \addpenalty\@secpenalty\addvspace{#2}%
    \begingroup \hyphenpenalty\@M
    \@ifempty{#4}{%
      \@tempdima\csname r@tocindent\number#1\endcsname\relax
    }{%
      \@tempdima#4\relax
    }%
    \parindent\z@ \leftskip#3\relax \advance\leftskip\@tempdima\relax
    \rightskip\@pnumwidth plus1em \parfillskip-\@pnumwidth
    #5\leavevmode\hskip-\@tempdima{#6}\nobreak
    \leaders\hbox{$\m@th\mkern \@dotsep mu\hbox{.}\mkern \@dotsep mu$}\hfill
    \nobreak
    \hbox to\@pnumwidth{\@tocpagenum{\ifnum#1=1\bfseries\fi#7}}\par
    \nobreak
    \endgroup
  \fi}
\AtBeginDocument{%
\expandafter\renewcommand\csname r@tocindent0\endcsname{0pt}
}
\def\l@subsection{\@tocline{2}{0pt}{2.5pc}{5pc}{}}
\makeatother

\usepackage{etoolbox}
\makeatletter
\patchcmd{\@setaddresses}{\indent}{\noindent}{}{}
\patchcmd{\@setaddresses}{\indent}{\noindent}{}{}
\patchcmd{\@setaddresses}{\indent}{\noindent}{}{}
\patchcmd{\@setaddresses}{\indent}{\noindent}{}{}
\makeatother

\usepackage[linktocpage]{hyperref}
\hypersetup{
  colorlinks   = true, 
  urlcolor     = blue, 
  linkcolor    = Purple, 
  citecolor   = red 
}
\usepackage{comment}
\usepackage{amscd}
\usepackage{mathtools}

\DeclareMathOperator{\C}{\mathcal{C}}
\newcommand{\srk}{\mathrm{srk}}
\newcommand{\dsrk}{\mathrm{d}_{\mathrm{srk}}}

\DeclareMathOperator{\Aut}{Aut}
\DeclareMathOperator{\End}{End}
\DeclareMathOperator{\Gal}{Gal}

\DeclareMathOperator{\rk}{rk}

\DeclareMathOperator{\Mat}{Mat}
\DeclareMathOperator{\Ext}{Ext}

\DeclareMathOperator{\GL}{GL}

\DeclareMathOperator{\ww}{w}
\usepackage{cleveref}

\theoremstyle{definition}
\newtheorem{theorem}{Theorem}[section]

\newtheorem{corollary}[theorem]{Corollary}
\newtheorem{definition}[theorem]{Definition}
\newtheorem{proposition}[theorem]{Proposition}

\newtheorem{example}[theorem]{Example}

\newtheorem{remark}[theorem]{Remark}

\newcommand{\fqn}{\mathbb{F}_{q^n}}

\newcommand{\F}{{\mathbb F}}

\newcommand{\Z}{{\mathbb Z}}
\newcommand{\NN}{{\mathbb N}}

\newcommand{\w}{{\mathrm w}}

\newcommand{\xx}{{\mathbf x}}
\newcommand{\yy}{{\mathbf y}}

\newcommand{\bfn}{\mathbf {n}}
\newcommand{\bfm}{\mathbf {m}}
\newcommand{\gcrd}{\mathrm{gcrd}}

\newcommand{\fq}{{\mathbb F}_{q}}
\newcommand{\Fq}{{\mathbb F}_{q}}
\newcommand{\Fm}{{\mathbb F}_{q^m}}
\newcommand{\K}{{\mathbb K}}

\newcommand{\N}{\mathrm{N}}

\newcommand{\spacmn}{\Mat(\bfm,\bfn,\F_q)}

\newcommand{\Fmnk}{[\bfn,k]_{q^m/q}}
\newcommand{\Fmnkd}{[\bfn,k,d]_{q^m/q}}

\newcommand{\lid}{\mathcal{I}_{\ell}}
\newcommand{\rid}{\mathcal{I}_{r}}

\newcommand{\st}{\,:\,}
\newcommand{\fqm}{\mathbb{F}_{q^m}}
\newcommand{\wl}{\mathrm{wt}_{\Lambda}}
\newcommand{\dl}{\mathrm{d}_{\Lambda}}
\newcommand{\RH}{R/RH_{\Lambda}}
\newcommand{\Hl}{H_{\Lambda}}
\newcommand{\Fix}{\mathrm{Fix}}
\usepackage{todonotes}
\DeclareRobustCommand{\intodo}[1]{%
  \todo[inline]{#1}%
}

\title{Invariants for sum-rank metric codes}
\date{}

\author[P. Santonastaso]{Paolo Santonastaso}
\address{Paolo Santonastaso, \textnormal{Dipartimento di Matematica e Fisica, Universit\`a degli Studi della Campania ``Luigi Vanvitelli'', Viale Lincoln, 5, I--\,81100 Caserta, Italy}\newline
\textnormal{Dipartimento di Meccanica, Matematica e Management, Politecnico di Bari, Via Orabona 4, 70125 Bari, Italy}}
\email{paolo.santonastaso@unicampania.it,paolo.santonastaso@poliba.it}

\author[F. Zullo]{Ferdinando Zullo}
\address{Ferdinando Zullo, \textnormal{Dipartimento di Matematica e Fisica, Universit\`a degli Studi della Campania ``Luigi Vanvitelli'', Viale Lincoln, 5, I--\,81100 Caserta, Italy}}
\email{ferdinando.zullo@unicampania.it}

\subjclass[2020]{12E10; 16S36; 94B60} 
\keywords{Sum-rank metric code, code invariant, nuclear parameter}
\begin{document}

\begin{abstract}
The code equivalence problem is central in coding theory and cryptography. While classical invariants are effective for Hamming and rank metrics, the sum-rank metric, which unifies both, introduces new challenges.
This paper introduces new invariants for sum-rank metric codes: generalised idealisers, the centraliser, the center, and a refined notion of linearity. These lead to the definition of nuclear parameters, inspired by those used in division algebra theory, where they are crucial for proving inequivalence.
We also develop a computational framework based on skew polynomials, which is isometric to the classical matrix setting but enables explicit computation of nuclear parameters for known MSRD (Maximum Sum-Rank Distance) codes. This yields a new and effective method to study the code equivalence problem where traditional tools fall short. In fact, using nuclear parameters, we can study the equivalence among the largest families of known MSRD codes.
\end{abstract}

\maketitle

\tableofcontents

\section{Introduction}

In coding theory, one of the fundamental questions is how to determine whether two linear codes are equivalent, that is, whether they can be transformed into one another by a certain isometry.
This problem, known as \emph{code equivalence problem}, plays a central role in the design and analysis of several cryptographic schemes.
In particular, many code-based cryptosystems rely on the assumed hardness of code equivalence under various metrics (e.g., Hamming or rank metric) to hide the structure of secret codes and ensure security; see, for example, \cite{barenghi2023computational,petrank2002code,reijnders2024hardness,biasse2025code}.
Classical invariants used to study code equivalence include the basic code parameters (length, dimension, and minimum distance), the weight distribution (i.e., the number of codewords of each weight), the structure of the automorphism group, and the supports of codewords, particularly those of minimum weight.
These invariants are preserved under code equivalence and are therefore essential tools in both theoretical classification and cryptanalytic attacks on code-based cryptosystems; see \cite{barenghi2023computational}.

This paper regards the study of invariants for sum-rank metric codes. 
The sum-rank metric unifies and generalizes both Hamming metric and rank metric, providing a flexible framework for applications in network coding, distributed storage, and cryptography; see e.g. \cite{martinez2022codes}. The codewords of a sum-rank metric code are vectors of matrices (or vectors over extension fields), and the distance between two codewords is defined as the sum of the ranks of the differences of their corresponding blocks.

Families such as Reed–Solomon codes and Gabidulin codes are examples or building blocks of optimal sum-rank metric codes.
These codes have recently gained significant attention in both theoretical research and practical applications, particularly in designing resilient multi-shot network codes and secure distributed storage systems; see \cite{martinezpenas2019reliable,martinezpenas2019universal}.

In recent years, a rich mathematical theory of sum-rank metric codes has been developed, primarily through the work of Martínez-Peñas \cite{Martinez2018skew,martinez2019theory,martinezpenas2021hamming}, and continues to grow with the emergence of new constructions of optimal codes tailored to various bounds; see \cite{byrne2021fundamental,abiad2023eigenvalue,ott2021bounds,abiad2024linear,camps2022optimal,alfarano2022sum}.
Indeed, within the framework of skew polynomials, the known constructions of MSRD codes—those that attain equality in the Singleton-like bound—are as follows:
\begin{itemize}
    \item Linearized Reed-Solomon (LRS) codes,
    \item (Additive) Twisted Linearized Reed-Solomon (ATLRS) codes,
    \item Twisted Linearized Reed-Solomon (TLRS) codes of TZ-type.
\end{itemize}

In most of the cases, these families provide examples of MSRD codes with the same parameters. This motivates the need to develop suitable invariants capable of distinguishing between such codes. Some of them have been recently developed for the sum-rank metric, such as supports \cite{martinez2019theory}, an Overbeck-like invariant \cite{hormann2022distinguishing}, and latroids \cite{gorla2025latroids}; however, they are generally very difficult to compute.
A similar situation occurred with the rank metric. A key idea was first introduced by Liebhold and Nebe in \cite{liebhold2016automorphism}, and later developed by Lunardon, Trombetti, and Zhou in \cite{lunardon2018nuclei}. They adapted classical invariants from the theory of semifields and division algebras to the rank-metric setting. The core idea is to associate certain algebras related to a rank-metric code that correspond to subgroups of its automorphism group—subgroups that are, in some sense, much easier to compute. These algebras are known as \emph{idealisers}.
They have been widely used to prove the inequivalence of certain MRD codes, such as Gabidulin codes, twisted Gabidulin codes \cite{trombetti2019nuclei}, and codes arising from a broad class of scattered polynomials \cite{longobardi2024scattered}.
In the first part of this paper, we study and characterize the notion of nondegenerate rank-metric codes, which have been mainly considered in the previous papers only for a class of rank-metric codes which are known as $\F_{q^m}$-linear codes; see e.g. \cite{Randrianarisoa2020ageometric,alfarano2022linear}.
We now introduce several new invariants for the sum-rank metric:

\begin{itemize}
\item generalised idealisers (Section~\ref{sec:genidea}),
\item the centraliser (Section~\ref{sec:central+center}),
\item the center (Section~\ref{sec:central+center}),
\item a refined notion of linearity (Section~\ref{sec:linearity}).
\end{itemize}

Generalised idealisers can be seen as an extension of classical idealisers from the rank metric. The key difference is that, in the sum-rank setting, these become sets of $t$-tuples of matrices. Moreover, there can be more than two idealisers - specifically, up to $2^t$, where $t$ corresponds to the number of blocks of a sum-rank metric code.
The notions of centraliser and center have been recently introduced in the rank-metric context as generalizations of the right nucleus and the center of semifields/division algebras (cf.\ \cite{sheekey2020new,thompson2023division}). In this work, we further extend these concepts to the sum-rank metric framework, where we prove that they serve as invariants for codes in this setting.
By analogy with the nuclear parameters of a semifield and those of a rank-metric code, we refer to the sizes of left and right idealisers, as well as those of the centraliser and the center, as the \emph{nuclear parameters} of a sum-rank metric code, since they are invariants under code equivalence.
Linearity is another important invariant. As shown in Example~\ref{ex:notlinear}, there exist equivalent sum-rank metric codes such that one is $\mathbb{F}_{q^m}$-linear and the other is only $\mathbb{F}_q$-linear, viewed as subspaces of the same ambient space. This implies that linearity alone is not sufficient to determine the code equivalence.
To resolve this, we introduce a refined notion of linearity based on the left idealiser. This approach allows us to distinguish codes more effectively and avoid the issue described above.
The second part of the paper is devoted to computing the nuclear parameters of the known families of MSRD codes.
To this end, we introduce the framework of skew polynomials (see \cite{neri2021twisted}), which is isometric to the classical setting. However, this approach provides a novel and powerful method that allows the explicit computation of nuclear parameters of known MSRD code families.
This perspective is new and opens up a completely different approach to proving the inequivalence of certain sum-rank metric codes — based not on traditional invariants, but on the structure and sizes of their nuclear parameters.
Indeed, thanks to the nuclear parameters, we are able to solve the equivalence issue among the largest three families of MSRD codes. We prove that an LRS code cannot be equivalent to an ATLRS code nor to a TLRS code of TZ-type. Moreover, TLRS codes of TZ-type are not equivalent to any TLRS code. Also, within the family of ATLRS codes, there are many inequivalent MSRD codes.

\section{Sum-rank metric Codes}

In this section, we will quickly recall definitions and properties of sum-rank metric codes in the two frameworks of matrices and skew group algebras.

\subsection{Matrix framework} The classical framework for sum-rank metric codes is based on the (external) direct sum of matrix spaces. So, each codeword is represented as a tuple of matrices over a finite field. The sum-rank weight of a codeword is defined as the sum of the ranks of its matrix components. We define it formally as follows.

Let $t$ be a positive integer. Let $m_1,\ldots,m_t,n_1,\ldots,n_t$ be positive integers.
We consider the product of $t$ matrix spaces 
$$\Mat(\textbf{m},\bfn,\F_q)= \bigoplus_{i=1}^t \F_q^{m_i \times n_i},$$
with $\mathbf{m}=(m_1,\ldots,m_t)$ and $\mathbf{n}=(n_1,\ldots,n_t)$. 
The \textbf{sum-rank distance} is the function 
\[
\dsrk : \Mat(\textbf{m},\bfn,\F_q) \times \Mat(\textbf{m},\bfn,\F_q) \longrightarrow \mathbb{N}=\{0,1,2,\ldots\}  
\]
defined by
\[
\dsrk(X,Y) = \sum_{i=1}^t \rk(X_i - Y_i),
\]
where $X=(X_1, \ldots, X_t)$, $Y=(Y_1,\ldots , Y_t)$, with $X_i,Y_i \in \F_q^{m_i \times n_i}$. We define the \textbf{sum-rank weight} of an element $X=(X_1,\ldots,X_t) \in \Mat(\textbf{m},\bfn,\F_q)$ as 
$$\w_{\srk}(X):=\sum_{i=1}^t \rk(X_i).$$
Clearly, $\dsrk(X,Y)= \w_{\srk}(X-Y)$, for every $X,Y \in \Mat(\textbf{m},\bfn,\F_q)$. \\
An \textbf{(additive) sum-rank metric code} $\mathcal{C}$ is an $\F_p$-linear subspace of $\Mat(\textbf{m},\bfn,\F_q)$ endowed with the sum-rank distance. We always suppose $\lvert \mathcal{C} \rvert \geq 2$.
The \textbf{minimum sum-rank distance} of a sum-rank metric code $\mathcal{C}$ is defined as usual via $$\dsrk(\mathcal{C})=\min\{\w_{\srk}(X): X \in \mathcal{C}, X \neq \mathbf{0}\}.$$ 


Clearly, when $t = 1$, the definition of sum-rank metric codes reduces to rank-metric codes, and when $m_i = n_i = 1$ for all $i$, it reduces to block codes of length $t$ equipped with the Hamming metric.

As for the classical Hamming metric codes, the parameters of a sum-rank metric code are related via a Singleton-like bound.

\begin{theorem} [see \textnormal{\cite[Theorem 3.2]{byrne2021fundamental}}] \label{th:SingletonboundmatrixRav}
    Let $\mathcal{C}$ be a sum-rank metric code in $\Mat(\textbf{m},\bfn,\F_q)$ with $m_i \leq n_i$, $m_1\geq \cdots \geq m_t$ and $n_1 \geq \cdots \geq n_t$ and $\dsrk(\mathcal{C})=d$.
    Let $M=m_1+\ldots+m_t$. 
    Let $j$ and $\delta$ be the unique integers satisfying $d-1=\sum_{i=1}^{j-1}m_i+\delta$ and $0 \leq \delta \leq m_j-1$. Then \[\lvert \mathcal{C} \rvert \leq q^{\sum_{i=j}^tm_in_i-n_j \delta}.\]
    If $n_1=\ldots=n_t=n$, then $|\C|\leq q^{n(M-d+1)}$.
\end{theorem}

The sum-rank metric codes whose parameters satisfy the equality in the above bound are called \textbf{Maximum Sum-Rank Distance} (or shortly \textbf{MSRD}) codes.

For sum-rank metric codes we can extend the notion of adjoint of rank-metric codes, with the adjoint operation acting blockwise.
Let $v=(v_1,\ldots,v_t)\in \Z_2^t$. The $v$-\textbf{adjoint} of a code $\C \in \Mat(\bfm,\bfn,\F_q)$ is the code
\[
\C^v=\{(X_1^{v_1},\ldots,X_t^{v_t}) \colon (X_1,\ldots,X_t) \in \C\},
\]
where 
\[
X_i^{v_i}:=\begin{cases}
    X_i & \mbox{ if }v_i=0, \\
    X_i^{\top} & \mbox{ if }v_i =1.
\end{cases}
\]
When $v=(1,\ldots,1)$, the code $\C^{\top}:=\C^v \subseteq \Mat(\bfn,\bfm,\F_q)$ is simply called the \textbf{adjoint code of }$\C$. Precisely, for any element $X=(X_1,\ldots,X_t) \in \Mat(\bfm,\F_q)$, we define
\[
X^{\top}:=(X_1^{1},\ldots,X_t^{1})=(X_1^{\top},\ldots,X_t^{\top}),
\]
so that $\C^{\top}=\{X^{\top} \colon X \in \C\}$.

Consider the nondegenerate bilinear form on $\Mat(\bfm,\bfn,\F_q)$ defined as 
\begin{equation} \label{eq:bilinearmatrix}
\sum_{i=1}^t\mathrm{Tr}_{q/p}(\mathrm{Tr}(X_iY_i^{\top})),
\end{equation}
for every $(X_1,\ldots,X_t),(Y_1,\ldots,Y_t) \in \Mat(\bfm,\bfn,\F_q)$. The \textbf{dual} $\C^{\perp}$ of a sum-rank metric code $\C$ in $\Mat(\bfm,\bfn,\F_q)$ is defined as 
\[
\C^{\perp}=\{ (Y_1,\ldots,Y_t) \in \Mat(\bfm,\bfn,\F_q) \colon \sum_{i=1}^t\mathrm{Tr}_{q/p}(\mathrm{Tr}(X_iY_i^{\top}))=0, \mbox{ for every }(X_1,\ldots,X_t) \in \C \}.
\]

A comprehensive study of the properties of codes in this framework can be found in \cite{byrne2021fundamental,byrne2022anticodes,camps2022optimal,martinez2022codes}.

\subsection{Skew group algebra framework} Let $\theta$ be a generator of $\Gal(\fqm/\fq)$. We denote by $\F_{q^m}[\theta]$ the skew group algebra
\[
\F_{q^m}[\theta]=\left\{\sum_{i=0}^{m-1}a_i\theta^i\colon a_i \in \F_{q^m}\right\}.
\]
The elements of $\F_{q^m}[\theta]$ can be also considered as polynomials in the indeterminate $\theta$ with coefficients in $\F_{q^m}$ and are also called \textbf{$\theta$-polynomials} (or \textbf{$\theta$-linearized polynomials}). 

In recent decades, the skew group algebra framework $\mathbb{F}_{q^m}[\theta]$ has proven to be fundamental in the study of rank-metric codes. This is due to the following reason.

It is well-known that the $\fq$-algebra $\F_{q^m}[\theta]$ is isomorphic to the $\F_q$-algebra $\End_{\F_q}(\F_{q^m})$ of the $\F_q$-endomorphisms of $\F_{q^m}$, see e.g. \cite[Theorem 1.3]{chase1969galois}. More precisely, for any element $f = \sum_{i=0}^{m-1}f_i\theta^i \in  \F_{q^m}[\theta]$, one can consider the map
\[
\begin{tabular}{l c c c }
$\phi_f:$ & $\F_{q^m}$ & $\longrightarrow$ & $\F_{q^m}$ \\
& $\beta$ & $\longmapsto$ & $\sum\limits_{i=0}^{m-1}f_i\theta^i(\beta) $.
\end{tabular}
\]
Then the map $\phi:f \mapsto \phi_f$ is an $\F_q$-algebra isomorphism between $\F_{q^m}[\theta]$ and $\End_{\F_q}(\F_{q^m})$. Therefore, we can identify $f$ with the map $\phi_f$, and so we will refer to $\rk(f)$ and $\ker(f)$ to indicate the rank over $\F_q$ of $\phi_f$ and its kernel, respectively. As an immediate consequence, we get that via $\phi$, the rank-metric space $(\F_q^{m \times m},\rk)$ is $\F_q$-linear isometric to the space $(\F_{q^m}[\theta],\rk)$. 

Therefore, rank-metric codes, originally introduced as sets of $n \times m$ matrices over a finite field $\mathbb{F}_q$, can particularly in the square case $n = m$, be studied as subspaces of linearized polynomials within a skew group algebra. So, the algebra $\mathbb{F}_{q^m}[\theta]$ provides an alternative and powerful framework for analyzing rank-metric codes. For a detailed exposition of this perspective, we refer the reader to \cite{sheekeysurvey, sheekey2016new}.

There is a natural generalization of this framework to the setting of sum-rank metric codes, cf. \cite[Section 2.5]{neri2021twisted}. Indeed, consider $t$ copies of the skew group algebra $\F_{q^m}[\theta]$. We define the \textbf{sum-rank weight} on $(\F_{q^m}[\theta])^t$ as 
\[
\w_{\srk}((f_1,\ldots,f_t)):=\sum_{i=1}^t \rk (f_i).
\]
This induces a metric on $(\F_{q^m}[\theta])^{t}$ defined as 
\begin{equation} \label{eq:sumrankpoly}
\dsrk((f_1,\ldots,f_t),(g_1,\ldots,g_t))=\w_{\srk}((f_1,\ldots,f_t)-(g_1,\ldots,g_t)).
\end{equation}

Since the space $(\F_q^{m \times m},\rk)$ is $\F_q$-linear isometric to the space $(\F_{q^m}[\theta],\rk)$ described above, we obtain that $((\F_{q^m}[\theta])^t,\dsrk)$ is isometric to the space $(\Mat(\bfm,\bfm,\F_q),\dsrk)$, with $\bfm=(m,\ldots,m)$, via the map
\begin{equation} \label{eq:ismotrymatrpoly}
    \begin{array}{l c c c }
\overline{\phi}:& (\F_{q^m}[\theta])^t & \longrightarrow & \Mat(\bfm,\bfm,\F_q) \\
& (f_1,\ldots,f_t) & \longmapsto & (\phi_{f_1},\ldots,\phi_{f_t}).
\end{array}
\end{equation}

Therefore, in the following, we will always consider the two frameworks, $(\mathbb{F}_{q^m}[\theta])^t$ and $\text{Mat}(\bfm,\bfm,\mathbb{F}_q)$, as equivalent. Thus, we will describe properties, examples, and constructions in one setting, but they can naturally be readapted to the other using the isomorphism $\overline{\phi}$.

\section{Additive isometries and nondegenerate sum-rank metric codes}

As traditionally done, we aim to identify the codes that can be transformed into one another via an isometry of the ambient space. To this end, we will first describe the additive isometry of $\text{Mat}(\mathbf{m},\mathbf{n},\F_q)$ and derive two definitions of equivalence among sum-rank metric codes. We will then introduce the concept of nondegeneracy for codes, which, in this context, means that the code cannot be isometrically embedded in a shorter code. Finally, we will establish some sufficient conditions for a code to be nondegenerate.

For our purposes, from now on, we will assume that for any $i$ 
\[\max\{m_i,n_i\}>1,\]
in order to avoid some criticisms in the characterization of isometries that we will see.

\subsection{Additive isometries}

The notion of code equivalence in the sum-rank metric was introduced in \cite[Theorem 2]{martinezpenas2021hamming}. Notably, the MacWilliams Extension Theorem does not hold even in the case $t = 1$ see \cite[Example 2.9]{barra2015macwilliams} and also \cite{gorla2024macwilliams}. As a result, equivalence of sum-rank metric codes is defined in terms of isometries of the entire ambient space.

In \cite{camps2022optimal}, the classification of $\F_q$-linear isometries on the space $\left( \Mat(\bfn,\bfm,\F_q),\dsrk \right)$ is provided. In the following, we use a more general notion of equivalence. In fact, we consider \emph{additive} isometries.

\begin{definition}
An \textbf{(additive) isometry} of the metric space $\left(\Mat(\bfm,\bfn,\F_q),\dsrk\right)$ is an additive bijective map $\varphi:\Mat(\bfm,\bfn,\F_q) \rightarrow \Mat(\bfm,\bfn,\F_q)$ that preserves the sum-rank metric distance, i.e.
\[
\dsrk(X,Y)=\dsrk(\varphi(X),\varphi(Y)),
\]
for each $X=(X_1,\ldots,X_t),Y=(Y_1,\ldots,Y_t) \in \Mat(\bfm,\bfn,\F_q)$. 
\end{definition}

Following the arguments of \cite[Theorem V.2]{camps2022optimal}, it is easy to prove the following classification for additive isometries of the space $(\Mat(\bfm,\bfn,\F_q),\dsrk)$, which relies on the classification of additive isometries of the rank-metric space $(\F_q^{m\times n},\rk)$. We recall that if $\psi:\F_q^{m\times n} \rightarrow \F_q^{m\times n}$ is an additive isometry, then there exist $A \in \GL(m,q), B \in \GL(n,q)$ such that 
$\phi(X)=AX^{\rho}B$, or $
\phi(X)=A(X^{\rho})^{\top}B$ if $m=n$,
for all $X \in \F_q^{m\times n}$, for some field automorphism $\rho$ of $\F_q$ acting entry-wise on $X$, see \cite{wan1996geometry}. 

\begin{theorem} \label{th:classificationadditiveisometries}
Let $\varphi:\Mat(\bfm,\bfn,\F_q) \rightarrow \Mat(\bfm,\bfn,\F_q)$ be an isometry of the metric space \\ $(\Mat(\bfm,\bfn,\F_q),\dsrk)$. Then there exists a permutation $\pi \in \mathcal{S}_t$, and there are rank-metric isometries $\psi_i:\F_q^{m_i\times n_i} \rightarrow \F_q^{m_i\times n_i}$, for every $i \in \{1,\ldots,t\}$, such that
\[
\varphi((X_1,\ldots,X_t))=(\psi_1(X_{\pi(1)}),\ldots,\psi_t(X_{\pi(t)}))
\]
for all $(X_1,\ldots,X_t) \in \Mat(\bfm,\mathbf{n},\F_q)$.
\end{theorem}

We say that two sum-rank metric codes $\C$ and $\C'$ in $\Mat(\bfm,\bfn,\F_q)$ are \textbf{(weakly) equivalent} if there exists an isometry $\varphi:\Mat(\bfm,\bfn,\F_q) \rightarrow \Mat(\bfm,\bfn,\F_q)$ such that $\varphi(\C)=\C'$.
In the next, we will consider only isometries $\varphi:\Mat(\bfm,\bfn,\F_q) \rightarrow \Mat(\bfm,\bfn,\F_q)$, for which $\psi_i:\Mat(m,n,\F_q) \rightarrow \Mat(m,n,\F_q)$ are of the form
$\phi(X)=AX^{\rho}B$ for some $A \in \GL(m,q), B \in \GL(n,q)$ and a field automorphism $\rho$ of $\F_q$ acting entry-wise on $X$. In other words, we will not take into account the transposition of the matrix in some blocks. Therefore, the notion of equivalence we will consider in the next is the following. We say that two sum-rank metric codes $\C$ and $\C'$ in $\Mat(\bfm,\bfn,\F_q)$ are \textbf{(strongly) equivalent} if there exists an isometry $\varphi:\Mat(\bfm,\bfn,\F_q) \rightarrow \Mat(\bfm,\bfn,\F_q)$, without any transposition, such that $\varphi(\C)=\C'$.

\subsection{Nondegenerate codes}\label{sec:nondeg}

The notion of nondegenerate sum-rank metric codes has been studied only in the case of $\mathbb{F}_{q^m}$-linear, even in the rank-metric case, cf. \cite{alfarano2022linear,neri2021geometry}. What we will say in the next still holds in the rank metric, and we will study the degeneracy condition in the more general case of additive codes.

\begin{definition}
A sum-rank metric code $\C \subseteq \mathrm{Mat}(\bfm,\bfn,\fq)$ is \textbf{degenerate} if there exists a sum-rank metric code $\C'$ in $\mathrm{Mat}(\bfm,\bfn,\fq)$ equivalent to $\C$ and an index $i \in \{1,\ldots,t\}$ for which the matrices of the $i$-th entry of all codewords of $\C'$ present (at least) one zero row or column in the same position.
\end{definition}

So, a degenerate sum-rank metric code can be isometrically embedded in a smaller space. Indeed, if a sum-rank metric code $\C$ is degenerate, then there exists an injective $\F_{q}$-linear map $\varphi: \C \rightarrow \mathrm{Mat}(\bfm',\bfn',\fq)$ with either $\bfm' < \bfm$ or $\bfn'< \bfn$ (where $<$ is the lexicographical order) such that $\w_{\srk}(X)=\w_{\srk}(\varphi(X))$, for each $X \in \C$.


\begin{example}
    Let 
    \[ \C=\left\langle\left( \begin{pmatrix} 
    1 &0 &0\\ 0 &1&0\\ 0 & 0 &0 
    \end{pmatrix} , (1)\right),  \left( \begin{pmatrix} 
    0 &0 &0\\ 0 &1&0\\ 0 & 0 &1 
    \end{pmatrix} , (1)\right)\right\rangle_{\F_2}\subseteq \mathrm{Mat}((3,1),(3,1),\F_2). \]
    It is easy to see that it is not possible to find $(A,a),(B,b) \in \mathrm{GL}(3,2)\times\F_2^*$ and a permutation $\pi \in S_2$ such that the induced isometry $\varphi$ sends $\C$ into a sum-rank metric code of $\mathrm{Mat}((3,1),(3,1),\F_2)$ for which either the matrices in the first entry of the codewords of $\varphi(\C)$ present a zero row/column in the same position or the second entry of the codeowords of  $\varphi(\C)$ is zero. Therefore, $\C$ is nondegenerate. Consider the code 
    \[ \C'=\left\langle\left( \begin{pmatrix} 
    1 &0 &1\\ 0 &1&1\\ 0 & 0 &0 
    \end{pmatrix} , (1)\right),  \left( \begin{pmatrix} 
    0 &0 &0\\ 0 &1&1\\ 1 & 0 &1 
    \end{pmatrix} , (1)\right)\right\rangle_{\F_2}\subseteq \mathrm{Mat}((3,1),(3,1),\F_2). \]
    As you can see, the third column of the matrices in the first block can always be obtained as the sum of the first two, therefore $\C'$ is equivalent to
    \[ \C''=\left\langle\left( \begin{pmatrix} 
    1 &0 &0\\ 0 &1&0\\ 0 & 0 &0 
    \end{pmatrix} , (1)\right),  \left( \begin{pmatrix} 
    0 &0 &0\\ 0 &1&0\\ 1 & 0 &0
    \end{pmatrix} , (1)\right)\right\rangle_{\F_2}\subseteq \mathrm{Mat}((3,1),(3,1),\F_2), \]
    and so it is degenerate. Note that $\C$ can be isometrically embedded in $\mathrm{Mat}((3,1),(2,1),\F_2)$ obtaining the following code
    \[ \overline{\C}=\left\langle\left( \begin{pmatrix} 
    1 &0\\ 0 &1\\ 0 & 0 
    \end{pmatrix} , (1)\right),  \left( \begin{pmatrix} 
    0 &0 \\ 0 &1\\ 1 & 0 
    \end{pmatrix} , (1)\right)\right\rangle_{\F_2}\subseteq \mathrm{Mat}((3,1),(2,1),\F_2). \]
\end{example}

We will start by characterizing the nondegeneracy condition for rank-metric codes and then we will describe it in the more general context of the sum-rank metric.

\begin{proposition}\label{prop:nondegrankmetric}
    Let $\C$ be a rank-metric code in $\F_q^{m\times n}$. Then $\C$ is nondegenerate if and only if 
    \begin{equation}\label{eq:nondegcond} 
    \bigcap_{C \in \C} \mathrm{leftker}(C)=\{\mathbf{0}\}\ \ \ \mbox{ and } \ \ \ \bigcap_{C \in \C} \mathrm{rightker}(C)=\{\mathbf{0}\}. 
    \end{equation}
\end{proposition}
\begin{proof}
    By contradiction, assume that $\C$ is nondegenerate and there exists a nonzero vector $\mathbf{v} \in\bigcap_{A \in \C} \mathrm{leftker}(A)$, that is a vector in $\mathbb{F}_q^m$ such that $\mathbf{v}C=\mathbf{0}$ for any $C \in \C$. Let $A$ be any matrix in $\mathrm{GL}(m,q)$ having as first row the vector $\mathbf{v}$, then the code $A\C$ is equivalent to $\C$ and all the matrices in $A\C$ have the zero row as first row, therefore $\C$ is degenerate.
    Similarly, one can argue in the case in which $\mathbf{v} \in\bigcap_{A \in \C} \mathrm{rightker}(A)$. Conversely, assume that \eqref{eq:nondegcond} holds and $\C$ is degenerate. Therefore, the exist $A \in \mathrm{GL}(m,q)$, $B \in \mathrm{GL}(n,q)$ and $\sigma \in \mathrm{Aut}(\F_q)$ such that the code $A\C^{\rho} B$ has all codewords with a fixed row or column equal to the zero vector. Without loss of generality, let us assume that all the codewords of $A\C^{\rho} B$ have the zero vector in the first row. Then it turns out that $A\C$ also has the same property. Denote by $\mathbf{v}$ the first row of $A$, then it turns out that $\mathbf{v}\in \mathrm{leftker}(C)$ for any $C \in \C$ and clearly $\mathbf{v}\ne\mathbf{0}$, a contradiction.
    If we assume that all codewords of $A\C^{\rho} B$ have the zero vector as the first column, we get the contradiction from the existence of a nonzero vector $\mathbf{v}\in \mathrm{rightker}(C)$.
\end{proof}

A sufficient condition that guarantees the nondegeneracy of a code in $\F_q^{m\times m}$ is to contain a matrix of full rank, although it is not an equivalent condition.

\begin{corollary}\label{cor:nondegrankmetric}
    Let $\C$ be a rank-metric code in $\F_q^{m\times m}$. 
    If $\C$ has a codeword of weight $m$ then $\C$ is nondegenerate.
\end{corollary}

\begin{remark}
The property of being nondegenerate for rank-metric codes has been characterized in the case where the code in $\mathbb{F}_q^{m \times m}$ is $\mathbb{F}_{q^m}$-linear (see Section~\ref{sec:linearity}), showing that this is equivalent to the code containing a codeword of maximum weight $m$ (see \cite[Proposition 3.11]{alfarano2022linear}). As observed in the above corollary, the existence of a codeword of full rank is a sufficient condition for nondegeneracy, but it is not a necessary one. Indeed, we can consider the rank-metric code constructed in \cite[Example 1]{dumas2011rank}: let $\C$ be the $\F_3$-linear rank-metric code in $\Mat(5,5,\F_3)$ spanned over $\F_3$ by
    \[\begin{pmatrix}
        2 & 2 & 2 & 1 & 2\\
        2 & 2 & 1 & 2 & 2\\
        2 & 1 & 0 & 1 & 1\\
        1 & 2 & 1 & 2 & 0\\
        2 & 2 & 1 & 0 & 0
    \end{pmatrix}, \begin{pmatrix}
        0 & 0 & 0 & 0 & 0\\
        0 & 0 & 0 & 0 & 1\\
        0 & 0 & 0 & 1 & 0\\
        0 & 0 & 1 & 1 & 0\\
        0 & 1 & 0 & 0 & 0
    \end{pmatrix}, \begin{pmatrix}
        0 & 0 & 0 & 1 & 1\\
        0 & 0 & 0 & 2 & 0\\
        0 & 0 & 0 & 2 & 2\\
        1 & 2 & 2 & 1 & 0\\
        1 & 0 & 2 & 0 & 0
    \end{pmatrix},  \]
    \[
    \begin{pmatrix}
        0 & 2 & 0 & 2 & 0\\
        2 & 1 & 2 & 0 & 0\\
        0 & 2 & 0 & 0 & 0\\
        2 & 0 & 0 & 1 & 1\\
        0 & 0 & 0 & 1 & 0
    \end{pmatrix},
    \begin{pmatrix}
        2 & 2 & 1 & 2 & 0\\
        2 & 2 & 1 & 2 & 1\\
        1 & 1 & 0 & 1 & 2\\
        2 & 2 & 1 & 2 & 2\\
        0 & 1 & 2 & 2 & 0
    \end{pmatrix}.
    \]
    As mentioned in \cite{dumas2011rank}, $\C$ ha dimension $5$ and all the nonzero codewords in $\C$ have rank $4$.
    Straightforward computations show that \eqref{eq:nondegcond} is satisfied and so $\C$ is nondegenerate, although it does not contain any full-rank matrix.
\end{remark}

\begin{remark}
    As expected, Corollary \ref{cor:nondegrankmetric} does not hold when $m \ne n$ (and by replacing the condition of having a matrix of rank $m$ with having a matrix of rank $\min\{m,n\}$). Indeed, consider
    \[ \C=\left\langle \begin{pmatrix}
    1 & 0 & 0\\
    0 & 1 & 0
\end{pmatrix}\right\rangle_{\F_2}, \]
we have that $\C$ contains a full-rank matrix but $\C$ is degenerate by Proposition \ref{prop:nondegrankmetric}. 
\end{remark}

We can now characterize the nondegeneracy of sum-rank metric codes in terms of rank-metric codes that we can associate to them, which we will refer to as components.

Consider the space $\mathrm{Mat}(\bfm,\bfn,\fq)$ and let $j \in \{1,\ldots,t\}$. Define the map
\[ \iota_j \colon \mathrm{Mat}(\bfm,\bfn,\fq) \mapsto \F_q^{m_j \times n_j}, \]
where $\iota_j((X_1,\ldots,X_t))=X_j$, that is $\iota_j$ gives the $j$-th entry of $(X_1,\ldots,X_t)$. We define the \textbf{$j$-th component} of a sum-rank metric code $\C \subseteq \mathrm{Mat}(\bfm,\bfn,\fq)$ as the rank-metric code
\[ \C_j=\iota_j(\C) \subseteq \F_q^{m_j \times n_j}. \]

In particular, the following properties hold:
\[ \C\subseteq \iota_1(\C)\oplus \ldots \oplus \iota_t(\C) \,\,\,\text{and}\,\,\, \ww_{\mathrm{srk}}(X)=\sum_{j=1}^{t} \mathrm{rk}(\iota_j(X)). \]

As a consequence of Proposition \ref{prop:nondegrankmetric}, a sum-rank metric code turns out to be nondegenerate if and only if all of its components are nondegenerate as rank-metric codes.

\begin{corollary}\label{cor:nondegsumrank}
 A sum-rank metric code $\C \subseteq \mathrm{Mat}(\bfm,\bfn,\mathbb{F}_q)$ is degenerate if and only if there exists an index $j \in \{1, \ldots, t\}$ such that the rank-metric code $\C_j = \iota_j(\C)$ is degenerate.
\end{corollary}

\begin{proof}
Assume first that the code $\C$ is degenerate. Then there exists an isometry 
\[
\varphi : \mathrm{Mat}(\bfm, \bfn, \F_q) \longrightarrow \mathrm{Mat}(\bfm, \bfn, \F_q)
\]
and a sum-rank metric code $\C' \subseteq \mathrm{Mat}(\bfm, \bfn, \F_q)$ equivalent to $\C$ via $\varphi$, such that there is an index $j \in \{1, \ldots, t\}$ for which the $j$-th block of all codewords in $\C'$ contains (at least) one row or column that is identically zero in the same position. In particular, the rank-metric code $\C_j' = \iota_j(\C')$ is degenerate. Since $\varphi$ is an isometry, there exists a permutation $\pi \in \mathcal{S}_t$ and rank-metric isometries $
\psi_j : \F_q^{m_j \times n_j} \longrightarrow \F_q^{m_j \times n_j}, \quad \text{for all } j \in \{1, \ldots, t\}$, such that $\varphi((X_1, \ldots, X_t)) = (\psi_1(X_{\pi(1)}), \ldots, \psi_t(X_{\pi(t)})),$ for all $(X_1, \ldots, X_t) \in \mathrm{Mat}(\bfm, \bfn, \F_q)$. Hence,
\[
\begin{array}{rl}
\C_j' &= \iota_j(\C') \\
&= \iota_j(\varphi(\C)) \\
&= \psi_j(\C_{\pi(j)}).
\end{array}
\]
Since $\C_j'$ is degenerate and $\varphi_j$ is an isometry, it follows that $\C_{\pi(j)}$ is also degenerate. Conversely, if there exists $j \in \{1,\ldots,t\}$ such that $\iota_j(\C)$ is degenerate, then $\iota_j(\C)$ is equivalent to a rank-metric code whose elements have a zero row or a zero column in the same position. We can extend this isometry to $\C$ and we obtain the code $\C'$ such that the $j$-th block of each codeword of $\C'$ has a zero row or a zero column in the same position, again we obtain that $\C$ is degenerate.
\end{proof}

Using Corollary \ref{cor:nondegrankmetric} we also have the following.

\begin{corollary}
    A sum-rank metric code $\C \subseteq \mathrm{Mat}(\bfm,\bfm,\fq)$ containing a codeword of weight $\sum_{i \in [t]}m_i$ is nondegenerate.
\end{corollary}
\begin{proof}
    By assumption $\C$ contains a codeword $(C_1,\ldots,C_t)$ of weight $\sum_{i \in [t]}m_i$, which implies that 
    \[ \mathrm{rk}(C_i)=m_i, \]
    for every $i$. Hence, any component of $\C$ has a full-weight codeword. Corollary \ref{cor:nondegsumrank} implies the assertion.
\end{proof}

Codewords with weight as in the above corollary are called \textbf{maximum weight codewords}, for a detailed study on them for the rank-metric codes we refer to \cite{polverino2024maximum}.
As a consequence of the above corollary and the following proposition, we can prove that, under certain assumptions, MSRD codes are nondegenerate.

\begin{proposition} [see \textnormal{\cite[Proposition A.5]{santonastaso2023msrd}}] \label{prop:numberweightMSRD} 
    Assume that $m=m_1=\ldots=m_t$ and $n=n_1=\ldots=n_t$ and $t\leq q-1$. Let $\C$ be an MSRD code in $\Mat(\bfm,\bfn,\F_q)$ having minimum distance $d$ and denote by $W_j(\C)$ the number of codewords in $\C$ having weight $j \in \{0.\ldots,t\min\{n,m\}\}$.  For any $0 \leq i \leq t\min\{n,m\}-d$, we have \[W_{d+i}(\C)>0,\] i.e. there always exists at least one codeword $X \in \C$ having weight $d+i$.  
\end{proposition}

The above proposition generalizes the result \cite[Lemma 52]{ravagnani2016rank}.
Such a result guarantees that an MSRD code has a maximum weight codeword. Therefore, such MSRD codes are nondegenerate.

\begin{corollary}\label{cor:MSRDcodesarenondeg}
    Let $\mathbf{m}=(m,\ldots,m)$, with $t \leq q-1$.
    Let $\C$ be an MSRD sum-rank metric code in $\Mat(\bfm,\bfm,\F_q)$.
    We have that $\C$ is nondegenerate.
\end{corollary}

Note that, the MSRD codes we will consider in Section \ref{sec:skewpolSRC} are nondegenerate codes.


\section{Invariants from rank metric and generalised idealisers}

This section is devoted to developing the theory of certain invariants for sum-rank metric codes. The main idea is to start from the invariants widely used in the rank metric and determine their sum-rank metric analogs. We begin with the classical notion of an idealiser, which we associate via viewing a sum-rank metric code as a rank-metric code (with a block structure). We extend this notion by considering possible matrix multiplication on the left or on the right in each block, introducing the notion of generalised idealisers. Then, we introduce the sum-rank metric analogs of the centralizer and center, which allow us to define the nuclear parameters of a sum-rank metric code. Finally, we conclude the section by clarifying that the notion of linearity, as it currently stands, cannot be used as an invariant. We propose a new notion of linearity that is preserved by the equivalence of sum-rank metric codes. We emphasize that, although the definitions and statements are expressed in terms of $t$-ple of matrices, they can be formulated in any framework isometric to $\mathrm{Mat}(\mathbf{m}, \mathbf{n}, \mathbb{F}_q)$. 

\subsection{Block matrix rank-metric codes as sum-rank metric codes}

As already observed in \cite[Section II]{byrne2021fundamental}, a sum-rank metric code in $\mathrm{Mat}(\bfm,\bfn,\fq)$ can also be seen as a rank-metric code whose codewords are block matrices. More precisely, define 
\[ M = \sum_{i=1}^{t} m_i \ \ \ \text{ and }\ \ \ N = \sum_{i=1}^{t} n_i.\] 
Recall that the direct sum of $A \in \mathrm{Mat}(m_1,n_1,\fq)$ and $B \in \mathrm{Mat}(m_2,n_2,\fq)$ is the block diagonal matrix
\[ A\oplus B=\begin{pmatrix}
    A & 0 \\
    0 & B
\end{pmatrix} \in \mathrm{Mat}(m_1+m_2,n_1+n_2,\F_q), \]
and we can define this operation iteratively on more than two matrices.
Consider $\psi$ from $\mathrm{Mat}(\bfm,\bfn,\fq)$ to $\mathrm{Mat}(M,N,\fq)$ such that
\[ \psi((X_1,\ldots,X_t))=X_1\oplus \ldots \oplus X_t. \]
The map $\psi$ is a linear map and consider $\mathcal{B}(\bfm,\bfn,\fq)$ as the image of $\psi$.
Also, note that 
\[\w_{\srk}((X_1,\ldots,X_t))=\mathrm{rk}(X_1\oplus \ldots \oplus X_t),\] 
for every $(X_1,\ldots,X_t) \in \Mat(\bfm,\bfn,\F_q)$.
and therefore $\psi$ is an isometry between the metric spaces $(\mathrm{Mat}(\bfm,\bfn,\fq),\dsrk)$ and $(\mathcal{B}(\bfm,\bfn,\fq),\mathrm{rk})$.

We can now consider some of the invariants known in the rank metric and see whether these can be extended to invariants in the sum-rank metric, as well. The first invariants were the idealisers of a code introduced by Liebhold and Nebe in \cite{liebhold2016automorphism}.
Given a rank-metric code $\C$ in $\mathrm{Mat}(M,N,\fq)$, we define 
the sets 
\[ L(\C)=\{ A \in \mathrm{Mat}(M,M,\fq) \colon AX \in \mathcal{C}, \,\,\, \text{for every}\,\,\,  X \in \C \} \]
\[ R(\C)=\{ B \in \mathrm{Mat}(N,N,\fq) \colon XB \in \mathcal{C}, \,\,\, \text{for every}\,\,\,  X \in \C \} \]
are called the \textbf{left idealiser} of $\C$ and the \textbf{right idealiser} of $\C$, respectively. 

These notions extend the classical concept of left and middle nuclei defined for division algebras and semifields (cf.~\cite{lunardon2018nuclei,sheekey2020new,lobillo2025quotients,thompson2023division}). These nuclei play a fundamental role in distinguishing between different semifields, as they are invariants under isotopy. 

Their extension to the context of rank-metric codes has proven to be particularly fruitful. Indeed, together with the operations of sum and product of matrices, both the left and the right idealisers are rings. Also, in \cite{lunardon2018nuclei}, the authors observed that if two rank-metric codes are equivalent, then their left/right idealisers are isomorphic as rings. 
In the next proposition, we will see that the idealisers of the rank-metric code associated to a sum-rank metric code have a block structure.

\begin{proposition}\label{prop:nondegleftrightideal}
    Let $\C$ be a nondegenerate sum-rank metric code in $\mathrm{Mat}(\bfm,\bfn,\fq)$. Then
    \[ L(\psi(\C))\subseteq \mathcal{B}(\mathbf{m},\mathbf{m},\F_q)\,\,\,\text{and}\,\,\, R(\psi(\C))\subseteq\mathcal{B}(\mathbf{n},\mathbf{n},\F_q). \]
\end{proposition}
\begin{proof}
    We will prove the assertion for the left idealiser, clearly similar arguments can be used in the case of the right idealiser.
    For a matrix $A \in \mathrm{Mat}(M,M,\fq)$, we have that $A \in L(\psi(\C))$ if and only if for any $C=(C_1,\ldots,C_t) \in \C$ there exists $C'=(C_1',\ldots,C_t') \in \C$ such that
    \[A \psi(C)=\psi(C').\]
    Write the matrix $A$ as a block matrix with $t^2$ blocks of the form $A_{i,j}\in \mathrm{Mat}(m_i,m_i,\F_q)$, where $i,j \in \{1,\ldots,t\}$. Using the block structure of $\psi(C)$ and $\psi(C')$ one gets 
    \[ A_{i,j}C_i=O, \]
    for any $i,j \in \{1,\ldots,t\}$ with $i\ne j$, where $O$ is the zero matrix in the corresponding ambient space.
    Let us fix $i,j \in \{1,\ldots,t\}$ with $i\ne j$. Then for any $C=(C_1,\ldots,C_t) \in \C$ we have 
    \[ A_{i,j}C_i=O, \]
    i.e. the rowspan of $A_{i,j}$ is contained in the left kernel of any element belonging to the $i$-th component of $\C$, that is $\iota_i(\C)$. Since $\mathcal{C}$ is nondegenerate, it follows from \Cref{cor:nondegsumrank} that the rank-metric code $\iota_i(\mathcal{C})$ is also nondegenerate, for every $i \in \{1,\ldots,t\}$. Then, by \Cref{prop:nondegrankmetric}, the row space of $A_{i,j}$ contains only the zero vector, implying that $A_{i,j}$ is the zero matrix for all $i, j \in \{1,\ldots,t\}$ with $i \ne j$.
 Therefore, $A=A_{1,1}\oplus \ldots \oplus A_{t,t}$ and so $A \in \mathcal{B}(\bfm,\bfm,\fq)$.
\end{proof}

By the result above, since $L(\psi(\mathcal{C})) \subseteq \mathcal{B}(\mathbf{m}, \mathbf{m}, \mathbb{F}_q)$ and $R(\psi(\mathcal{C})) \subseteq \mathcal{B}(\mathbf{n}, \mathbf{n}, \mathbb{F}_q)$, we can consider $\psi^{-1}(L(\psi(\mathcal{C})))$ and $\psi^{-1}(R(\psi(\mathcal{C})))$. As we will see in the next subsection, these sets are invariants also with respect to the sum-rank metric, and actually can be defined as follows.

\begin{definition}
 Let $\C$ be a sum-rank metric code in $\spacmn$. The \textbf{left idealiser} of $\C$ is 
 \[
 \lid(\C):=\left\{(D_1,\ldots,D_t) \in \Mat(\bfm,\bfm,\fq): (D_1X_1,\ldots,D_tX_t) \in \C, \mbox{ for every } (X_1,\ldots,X_t) \in \C\right\}.
 \]
 Similarly, the \textbf{right idealiser} of $\C$ is 
 \[
 \rid(\C):=\left\{(D_1,\ldots,D_t) \in \Mat(\bfn,\bfn,\fq): (X_1D_1,\ldots,X_tD_t) \in \C, \mbox{ for every } (X_1,\ldots,X_t) \in \C\right\};
 \]
\end{definition}

With the above definition in mind, Proposition \ref{prop:nondegleftrightideal} reads as follows: for nondegenerate sum-rank metric codes $\lid(\C)=\psi^{-1}(L(\psi(\C)))$ and $\rid(\C)=\psi^{-1}(R(\psi(\C)))$.
Therefore, some of the properties of $\lid(\C)$ and $\rid(\C)$ are inherited from $L(\psi(\C))$ and $R(\psi(\C))$.
Indeed, using \cite[Lemmas 4.3 and 4.4]{lunardon2018nuclei} and \cite[Theorem 3.1]{longobardi2024standard} we have that left and right idealisers turn out to be fields.

\begin{proposition}\label{prop:fieldideal}
 Let $\C$ be a sum-rank metric code in $\spacmn$ and assume that $m_i \leq n_i$, for each $i$ and let $d$ be its minimum distance. Suppose that $\C$ contains a codeword of weight $\sum_{i=1}^t m_i$. 
 \begin{itemize}
     \item If $d \geq \left\lfloor \frac{\sum_{i=1}^tm_i}2 \right\rfloor +1$ then $\lid(\C)$ is a subfield of $\F_{q^{m_1+\ldots+m_t}}$.
     \item If $d \geq \left\lfloor \frac{\sum_{i=1}^tn_i}2 \right\rfloor +1$ then $\rid(\C)$ is a subfield of $\F_{q^{n_1+\ldots+n_t}}$.
 \end{itemize}
\end{proposition}


We can use the Singleton-like bound to give an upper bound on the size of the idealisers when they are fields.
\begin{proposition}\label{prop:sizeideal}
 Let $\C$ be a MSRD code in $\spacmn$ and assume that $m_i \leq n_i$, for each $i$ and let $d$ be its minimum distance. 
 \begin{itemize}
     \item If $\lid(\C)$ is a field, then $|\lid(\C)|\leq q^{m_t}$. Moreover, if $m_1=\ldots=m_t=m$ then $|\lid(\C)|\leq q^m$.
     \item If $\rid(\C)$ is a field, then $|\lid(\C)|\leq q^{n_t}$. Moreover, if $n_1=\ldots=n_t=n$ then $|\lid(\C)|\leq q^n$.
 \end{itemize}
\end{proposition}
\begin{proof}
    Note that $\lid(\C) \subseteq \Mat(\bfm,\bfm,\fq)$ and, since $\lid(\C)$ is a field, we have that every nonzero element in $\lid(\C)$ has weight $m_1+\ldots+m_t$.
    Therefore, we can apply the Singleton-like bound (cf. Theorem \ref{th:SingletonboundmatrixRav}) on $\lid(\C)$ and then we obtain the statement (where $j=t$ and $\delta=m_t-1$ for $\lid(\C)$).
    Similar arguments can be performed on $\rid(\C)$ (in this case $\delta=n_t-1$).
\end{proof}

\subsection{Generalised idealisers}\label{sec:genidea}


In this subsection we will extend the notion of idealiser seen before.
Roughly speaking, we can define different types of idealisers according to whether we multiply the entries of the codewords on the right or on the left. Clearly, if we multiply all the entries of the codewords on the left, we obtain the notion of left idealiser, and if we multiply on the right, we obtain the right idealiser seen above.
To introduce them, we need some preliminary notation. 
Consider the space $\spacmn$. Let $\mathbf{s}\in \mathbb{Z}_2^t$ and consider
\[ \Pi_{\mathbf{s}}=\oplus_{i \in [t]} \mathrm{Mat}(s_i m_i+(1-s_i)n_i,s_i m_i+(1-s_i)n_i,\F_q), \]
that is we consider the direct sum of spaces of square matrices of order $m_i\times m_i$ if $s_i=1$ and of order $n_i\times n_i$ if $s_i=0$.
We can now define the following map
\[ \circ_{\mathbf{s}}\colon \Pi_{\mathbf{s}} \times \spacmn \rightarrow \spacmn, \]
and
\[ (A_1,\ldots,A_t)\circ_{\mathbf{s}}(C_1,\ldots,C_t)=(B_1,\ldots,B_t), \]
where
\[ B_i=\begin{cases}
A_iC_i & \text{if } s_i=1,\\
C_iA_i & \text{if } s_i=0.\\
\end{cases}\]
Clearly, if $\C$ is a subset of $\spacmn$ and $(A_1,\ldots,A_t) \in \Pi_{\mathbf{s}}$ then
\[ (A_1,\ldots,A_t)\circ_{\mathbf{s}}\C=\{ (A_1,\ldots,A_t)\circ_{\mathbf{s}}(C_1,\ldots,C_t) \colon (C_1,\ldots,C_t)\in \C \}. \]

We are now ready to define the generalised idealisers.

\begin{definition}
    Let $\C$ be a sum-rank metric code in $\spacmn$ and let $\mathbf{s}\in \mathbb{Z}_2^t$. The \textbf{$\mathbf{s}$-idealiser} of $\C$ is 
 \[
 \mathcal{I}_{\mathbf{s}}(\C):=\left\{(A_1,\ldots,A_t) \in \Pi_{\mathbf{s}} \colon (A_1,\ldots,A_t)\circ_{\mathbf{s}}\C \subseteq \C \right\}.
 \]
\end{definition}

We note that, in the case where $\mathbf{s}=(1,\ldots,1)$ then $\mathcal{I}_{\mathbf{s}}(\C)=\mathcal{I}_\ell(\C)$ and if $\mathbf{s}=(0,\ldots,0)$ then $\mathcal{I}_{\mathbf{s}}(\C)=\mathcal{I}_r(\C)$.
We show an example in terms of skew polynomials as it is easier to describe.

\begin{example}
    Let $m \in \NN$ be an even number. Let $\theta \colon \alpha \in \F_{q^m} \mapsto \alpha^q \in \F_{q^m}$.
    In the skew group algebra $\fqm[\theta]$, we consider $f=\mathrm{id}+\theta+\ldots+\theta^{m-1}$ and $g=\mathrm{id}+\theta^2+\ldots+\theta^{m-2}$.
    Consider $\C=\langle f\rangle_{\F_{q^{m}}}\oplus \langle g\rangle_{\F_{q^{m/2}}} \subseteq (\F_{q^m}[\theta])^2$.
    Observe that 
    \[R(\langle f\rangle_{\F_{q^{m}}})=\{ (\alpha-f_{m-1}^q-\ldots-f_1^{q^{m-1}})\mathrm{id}+f_1\theta+\ldots+f_{m-1}\theta^{m-1} \colon\]\[ \alpha \in \fq, f_1,\ldots,f_{m-1} \in \F_{q^{m}} \}\]
    and 
    \[R(\langle g\rangle_{\F_{q^{m/2}}})=\{ (\alpha-g_{m/2-1}^{q^2}-\ldots-g_2^{q^{m/2-1}})\mathrm{id}+(-g_{m-1}^{q^2}-g_{m-3}^{q^4}-\ldots-g_3^{q^{m-2}})\theta+g_2\theta^2+\ldots+g_{m-1}\theta^{m-1} \colon \]\[ \alpha \in \F_{q^2}, g_2,\ldots,g_{m-1} \in \F_{q^{m}} \}.\]
    Therefore, their sizes are 
    \[ |R(\langle f\rangle_{\F_{q^{m}}})|=q^{1+m(m-1)} \,\,\,\text{and}\,\,\, |R(\langle g\rangle_{\F_{q^{m/2}}})|=q^{2+m(m-2)}. \] 
    Moreover, $L(\langle f\rangle_{\F_{q^{m}}})\simeq \F_{q^{m}}$ and $L(\langle g\rangle_{\F_{q^{m/2}}})\simeq \F_{q^{m/2}}$.
Straightforward calculations show that
    \begin{itemize}
        \item $\mathcal{I}_{\ell}(\C) = L(\langle f\rangle_{\F_{q^{m}}}) \oplus L(\langle g\rangle_{\F_{q^{m/2}}}) $;
        \item $\mathcal{I}_{r}(\C) = R(\langle f\rangle_{\F_{q^{m}}}) \oplus R(\langle g\rangle_{\F_{q^{m/2}}}) $;
        \item $\mathcal{I}_{(1,0)}(\C)=L(\langle f\rangle_{\F_{q^{m}}}) \oplus R(\langle g\rangle_{\F_{q^{m/2}}}) $;
        \item $\mathcal{I}_{(0,1)}(\C)=R(\langle f\rangle_{\F_{q^{m}}}) \oplus L(\langle g\rangle_{\F_{q^{m/2}}}) $.
    \end{itemize}
Comparing their sizes, it is easy to see that the idealisers $\mathcal{I}_{\ell}(\C), \mathcal{I}_{r}(\C), \mathcal{I}_{(1,0)}(\C)$ and $\mathcal{I}_{(0,1)}(\C)$ are not isomorphic.
\end{example}

As it happens for the idealisers in the rank-metric, we can equip $\mathcal{I}_{\mathbf{s}}(\C)$ with entry-sum of matrices, entry-wise product of matrices obtaining that $\mathcal{I}_{\mathbf{s}}(\C)$ is a ring.

\begin{proposition}\label{prop:equivcodeimpliesequivideal}
    Let $\C_1$ and $\C_2$ be sum-rank metric codes in $\spacmn$ and let $\mathbf{s}\in \mathbb{Z}_2^t$. If $\C_1$ and $\C_2$ are equivalent, then $\mathcal{I}_{\mathbf{s}}(\C_1)$ and $\mathcal{I}_{\mathbf{s}}(\C_2)$ are isomorphic.
\end{proposition}
\begin{proof}
    By hypothesis, $\C_1$ and $\C_2$ are equivalent, so there exist $(A_1,\ldots,A_t) \in \Mat(\bfm,\bfm,\K)$, $(B_1,\ldots,B_t) \in \Mat(\bfn,\bfn,\K)$, $(\sigma_1,\ldots,\sigma_t) \in (\Aut(\K))^t$ and a permutation $\pi \in \mathcal{S}_t$ such that 
    \[
    \C_2=\{(A_1X_{\pi(1)}^{\sigma_1}B_1,\ldots,A_tX_{\pi(t)}^{\sigma_t}B_t)  \colon (X_1,\ldots,X_t) \in \C_1\}.
    \]
    Now, note that an element $(D_1,\ldots,D_t) \in \Pi_{\mathbf{s}}$ belongs to the idealisers $I_{\mathbf{s}}(\C_1)$ if and only if $(E_1,\ldots,E_t)$ belongs to $I_{\mathbf{s}}(\C_2)$, where
   \begin{equation}\label{eq:Ei's} E_i=\begin{cases}
    A_{i} D_{\pi(i)} A_{i}^{-1}, & \text{if } s_{i}=1,\\
    B_{i}^{-1} D_{\pi(i)} B_{i}, & \text{if } s_{i}=0.
    \end{cases}\end{equation} 
    This means that
    \[
    I_{\mathbf{s}}(\C_2)= \left\{(E_1,\ldots,E_t) \colon \text{where the } E_i\text{'s are defined as in \eqref{eq:Ei's}}\right\},
    \]
    implying that $I_{\mathbf{s}}(\C_2)$ and $I_{\mathbf{s}}(\C_2)$ are isomorphic as rings.
\end{proof}

The above result shows that the idealisers can be used as a distinguisher for sum-rank metric codes.

\begin{corollary}\label{cor:codeequiv->idealequiv}
    Let $\C_1$ and $\C_2$ be sum-rank metric codes in $\spacmn$. If there exists $\mathbf{s}\in \mathbb{Z}_2^t$ such that $\mathcal{I}_{\mathbf{s}}(\C_1)$ and $\mathcal{I}_{\mathbf{s}}(\C_2)$ are not isomorphic then $\C_1$ and $\C_2$ are not equivalent.
\end{corollary}

\begin{remark}
Note that if $\C$ is also $\F_{p^i}$-linear then the left idealiser contains a subring isomorphic to $\F_{p^i}$. This is due to the fact that one can consider the set $\mathcal{M}$ of the matrices  $M_{\alpha} \in \Mat(\bfm,\bfm,\fq)$ with $\alpha \in \F_{p^i}$, where \[
M_{\alpha}=\alpha(I_{m_1},\ldots,I_{m_t}).
\] Clearly, $\mathcal{M}\subseteq \mathcal{I}_{\ell}(\C)$.  Moreover, the code $\C$ can be seen as a left $\lid(\C)$-module and a right $\rid(\C)$-module.
\end{remark}

Concerning the effects of $\mathbf{v}$-adjoint and dual operation, we have these results on idealisers.

\begin{proposition} \label{prop:dualtranspideal}
Let $\C$ be an additive sum-rank metric code in $\spacmn$. 
For any $\mathbf{s},\mathbf{v}\in \mathbb{Z}_2^t$ we have that
\[ \mathcal{I}_{\mathbf{s}}(\C^{\mathbf{v}})=\mathcal{I}_{\mathbf{s}-\mathbf{v}}(\C)^{\mathbf{v}} \text{ and } \mathcal{I}_{\mathbf{s}}(\C^{\perp})=\mathcal{I}_{\mathbf{s}}(\C)^{\mathbf{1}}, \]
where $\mathbf{1}=(1,\ldots,1)\in \Z_2^t$.
\end{proposition}
\begin{proof}
For the first equality you need just to observe that if $(A_1,\ldots,A_t) \in \mathcal{I}_{\mathbf{s}}(\C^{\mathbf{v}})$, then for any $(C_1,\ldots,C_t)\in \C$ we have
\[ (A_1,\ldots,A_t)\circ_{\mathbf{s}}(C_1,\ldots,C_t) \in \C, \]
and so
\[ ((A_1,\ldots,A_t)\circ_{\mathbf{s}}(C_1,\ldots,C_t))^{\mathbf{v}}=(A_1,\ldots,A_t)^{\mathbf{v}}\circ_{\mathbf{s}-\mathbf{v}}(C_1,\ldots,C_t)^\mathbf{v} \in \C^{\mathbf{v}}. \]
For the second equality, consider $(A_1,\ldots,A_t) \in \mathcal{I}_{\mathbf{s}}(\C)$. Then for any $(D_1,\ldots,D_t) \in \C^\perp$ and $(C_1,\ldots,C_t)\in \C$, using classical properties of the trace of a matrix, we have
\[ \sum_{i=1}^t \mathrm{Tr}_{q/p}(\mathrm{Tr}(C_i(A_i^\top \circ_{s_i} D_i)^\top))=\sum_{i=1}^t \mathrm{Tr}_{q/p}(\mathrm{Tr}((A_i\circ_{s_i} C_i)D_i^\top))=0, \]
since $(A_1,\ldots,A_t)\circ_{\mathbf{s}}(C_1,\ldots,C_t) \in \C$.
Therefore, $(A_1,\ldots,A_t)^{\mathbf{1}} \in \mathcal{I}_{\mathbf{s}}(\C^\perp)$ and so $\mathcal{I}_{\mathbf{s}}(\C)^{\mathbf{1}} \subseteq \mathcal{I}_{\mathbf{s}}(\C^{\perp})$.
Using the same argument on $\C^\perp$, we obtain $\mathcal{I}_{\mathbf{s}}(\C^{\perp})^{\mathbf{1}} \subseteq \mathcal{I}_{\mathbf{s}}(\C)$, from which we can derive the equality.
\end{proof}

For MRD codes, under some conditions of the ambient space, it has been shown that their left and right idealisers are fields. This mirrors a well-known feature of division algebras, where both the left and middle nuclei are fields. In what follows, we prove that the generalised idealisers of MSRD codes are also fields as well.  

We first analyse the case of left and right idealisers.

\begin{proposition}
    Let $\C$ be a sum-rank metric code in $\spacmn$ and assume that $n=n_1=\ldots=n_t$, $m=m_1=\ldots=m_t$, $m \leq n$ and $t\leq q-1$, and let $d>1$ be its minimum distance. We have that $\lid(\C)$ is a field.
    Moreover, if $\max\{ d, tm-d+2 \}\geq \left\lfloor \frac{tn}2 \right\rfloor +1$ then $\rid(\C)$ is a field.
\end{proposition}
\begin{proof}
   By Proposition \ref{prop:numberweightMSRD}, there exists a codeword in $\C$ of weight $tm$.
    Therefore, if $d \geq \left\lfloor \frac{tm}2 \right\rfloor +1$, $\lid(\C)$ is a field according to Proposition \ref{prop:fieldideal}.
    Suppose now that $d < \left\lfloor \frac{tm}2 \right\rfloor +1$. By \cite[Theorem VI.I]{byrne2021fundamental} (see also \cite[Theorem 5]{martinez2019theory}), the dual $\C^\perp$ of $\C$ is still an MSRD code and
    \[ d(\C^\perp)=tm-d+2. \]
    Since $d < \left\lfloor \frac{tm}2 \right\rfloor +1$, we obtain that 
    \[ d(\C^\perp) \geq \left\lfloor \frac{tm}2 \right\rfloor +1. \]
    Hence we can use Proposition \ref{prop:fieldideal} to state that $\lid(\C^\perp)$ is a field.
    We note that, by Proposition \ref{prop:dualtranspideal}, $\lid(\C^\perp)=\lid(\C)^{\mathbf{1}}$ and so $\lid(\C)$ is a field as well.
    Using the same argument, we can prove the second part of the statement.
\end{proof}

\begin{remark}
    Clearly, if $\C$ is an MSRD code in $\spacmn$ with minimum distance one, then $\C=\spacmn$ and its left and right idealisers correspond to $\Mat(\bfm,\bfm,\F_q)$ and $\Mat(\bfn,\bfn,\F_q)$, respectively.
\end{remark}

As a consequence, when $m=n$ we have that the right and left idealisers of are both fields.

\begin{corollary}\label{cor:lridealfields}
    Let $\C$ be an MSRD code in $\Mat(\bfm,\bfm,\F_q)$ and assume that $m=m_1=\ldots=m_t$ and $t\leq q-1$, and let $d>1$ be its minimum distance. We have that $\lid(\C)$ and $\rid(\C)$ are fields.
\end{corollary}

We can now prove that the generalised idealisers of MSRD codes are fields, as well.

\begin{corollary}
    Let $\C$ be an MSRD code in $\Mat(\bfm,\bfm,\F_q)$ and assume that $m=m_1=\ldots=m_t$ and $t\leq q-1$, and let $d>1$ be its minimum distance. We have that $\mathcal{I}_{\mathbf{s}}(\C)$ is a field for every $\mathbf{s} \in \Z_2^{t}$.
\end{corollary}
\begin{proof}
    Let $\mathbf{s} \in \Z_2^{t}$ and note that if $\mathbf{v}=\mathbf{1} -\mathbf{s}$, then
    \[ \mathcal{I}_{\mathbf{s}}(\C)=\mathcal{I}_{\mathbf{s}-\mathbf{v}}(\C^{\mathbf{v}})^{\mathbf{v}}=\mathcal{I}_{\ell}(\C^{1-\mathbf{s}})^{1-\mathbf{s}}, \]
    by Proposition \ref{prop:dualtranspideal}.
    Since $\C^{\mathbf{1}-\mathbf{s}}$ is an MSRD code as well, we can apply Corollary \ref{cor:lridealfields} to show that $\mathcal{I}_{\ell}(\C^{\mathbf{1}-\mathbf{s}})$ is a field. As a consequence, $\mathcal{I}_{\ell}(\C^{\mathbf{1}-\mathbf{s}})^{\mathbf{1}-\mathbf{s}}$ is a field as well, implying our assertion.
\end{proof}

\subsection{Centraliser and center}\label{sec:central+center}

From the theory of rank-metric codes, we can introduce two additional invariants: the \emph{centraliser} and the \emph{center} of a sum-rank metric code. These concepts were introduced in \cite{sheekey2020new} as invariants for rank-metric codes, extending the notions of the right nucleus and center in the context of semifields and division algebras (see also \cite{thompson2023division}). In the following definition, we further extend these notions to the sum-rank metric setting. 

\begin{definition}
 Let $\C$ be a sum-rank metric code in $\Mat(\bfm,\bfm,\F_q)$. The \textbf{centralizer} of $\C$ is defined as
\[ \mathrm{C}(\C)=\{ (D_1,\ldots,D_t) \in \Mat(\bfm,\bfm,\F_q) \colon D_iA_i=A_iD_i, \mbox{ for every }i \in \{1,\ldots,t\}, (A_1,\ldots,A_t)\in \C \}. \]
The \textbf{center} of $\C$ is defined as
\[ Z(\C)=\lid(\C)\cap \mathrm{C}(\C). \]
\end{definition}

The centralizer and center of a code in 
$\Mat(\bfm,\bfm,\F_q)$ are subrings of $\Mat(\bfm,\bfm,\F_q)$. 

\begin{remark}
In the rank-metric setting, we observe that although the center of a code is defined as the intersection $Z(\C) = L(\C) \cap \mathrm{C}(\C)$, where $L(\C)$ and $\mathrm{C}(\C)$ denote the left idealiser and the centraliser of $\C$ respectively, it also holds that $Z(\C) = R(\C) \cap \mathrm{C}(\C)$, with $R(\C)$ being the right idealiser of $\C$. That is, the center can be equivalently described as the intersection between the right idealiser and the centraliser. A similar property holds in the sum-rank setting: it is easy to check that for any sum-rank metric code $\C$ in $\mathrm{M}(\bfm,\bfm,\F_q)$ and any $\mathbf{s} \in \Z_2^t$, we have
\[
Z(\C) = \mathcal{I}_{\mathbf{s}}(\C) \cap \mathrm{C}(\C).
\]
\end{remark}

We now prove that if two codes are equivalent, then their centralizers and centers are isomorphic as rings. In other words, both the centralizer and the center serve as invariants for sum-rank metric codes.

\begin{proposition} \label{prop:equivalencecenter}
    Let $\C$ and $\C'$ be equivalent sum-rank metric codes in $\Mat(\bfm,\bfm,\F_q)$, via an isometry $\varphi$, i.e. $\varphi(\C)=\C'$.  
    Assume that $(I_{m_1},\ldots,I_{m_t}) \in \C$ and $(I_{m_1},\ldots,I_{m_t}) \in \C'$.
    Let $\pi \in \mathcal{S}_t$ be a permutation and \[
    \psi_i:X_i \in  M(m_i,m_i,\F_q) \longrightarrow A_i X_i^{\rho_i} B_i \in  M(m_i,m_i,\F_q),\] be rank-metric isometries, with $A_i,B_i \in \GL_q(m_i)$, and $\rho_i\in \Aut(\F_q)$ for every $i \in \{1,\ldots,t\}$, such that
\[
\varphi((X_1,\ldots,X_t))=(\psi_1(X_{\pi(1)}),\ldots,\psi_t(X_{\pi(t)}))
\]
for all $(X_1,\ldots,X_t) \in \Mat(\bfm,\bfm,\F_q)$. We have 
\[
C(\C')=(A_1^{-\rho_1^{-1}},\ldots,A_t^{-\rho_t^{-1}})\circ_{\mathbf{1}} C(\pi(\C)) \circ_{\mathbf{1}} (A_1^{\rho_1^{-1}},\ldots,A_t^{\rho_t^{-1}})
\]
and 
\[
Z(\C')=(A_1^{-\rho_1^{-1}},\ldots,A_t^{-\rho_t^{-1}})\circ_{\mathbf{1}}  Z(\pi(\C)) \circ_{\mathbf{1}} (A_1^{\rho_1^{-1}},\ldots,A_t^{\rho_t^{-1}})
\]
In particular, their centralizers and centers are isomorphic.
\end{proposition}

\begin{proof}
    For an element $(D_1,\ldots,D_t) \in C(\C')=C(\varphi(\C))$, we have 
    \[
    (D_1,\ldots,D_t)\circ_{\mathbf{1}} (\psi_1(X_{\pi(1)}),\ldots,\psi_t(X_{\pi(t)}))= (\psi_1(X_{\pi(1)}),\ldots,\psi_t(X_{\pi(t)}))\circ_{\mathbf{1}} (D_1,\ldots,D_t),
    \]
    for every $(X_1,\ldots,X_t) \in \C$, implying 
    \[
    D_i A_i X_{\pi(i)}^{\rho_i}B_i=A_i X_{\pi(i)}^{\rho_i}B_i D_i, 
    \]
    for all $i\in \{1,\ldots,t\}$. 
    Thus,
    \begin{equation} \label{eq:permutationcentralisers}
    (A_i^{-1}D_i A_i)^{\rho_i^{-1}}X_{\pi(i)}=X_{\pi(i)}(B_i D_i B_i^{-1})^{\rho_i^{-1}}, 
    \end{equation}
    for every $i \in \{1,\ldots,t\}$.
    Since $(I_{m_1},\ldots,I_{m_t})$ is in $\C$, we have 
    \[
    (A_i^{-1}D_i A_i)^{\rho_i^{-1}}=(B_i D_i B_i^{-1})^{\rho_i^{-1}}, 
    \]
    and so, by \eqref{eq:permutationcentralisers}, we obtain $(A_i^{-1}D_i A_i)^{\rho_i^{-1}} \in C(\C)$. Applying the same arguments with the roles of $\C$ and $\C'$ swapped, the proof is complete.
    Finally, we prove the relation for the two centers. Recall that by Proposition \ref{prop:equivcodeimpliesequivideal}, see also its proof and \eqref{eq:Ei's}, 
    \[
    \lid(\C')=(A_1^{-\rho_1^{-1}},\ldots,A_t^{-\rho_t^{-1}}) \circ_{\mathbf{1}}\lid(\pi(\C)) \circ_{\mathbf{1}}(A_1^{\rho_1^{-1}},\ldots,A_t^{\rho_t^{-1}}),
    \]
    and that $Z(\C')=\lid(\C') \cap C(\C')$, we have the assertion. 
\end{proof}

The above proposition extends \cite[Proposition 4]{sheekey2020new}, which was originally proven in the context of rank-metric codes.

\begin{remark}
We note that a condition for applying \Cref{prop:equivalencecenter} is that both codes $\C$ and $\C'$ contain the element $(I_{m_1}, \ldots, I_{m_t})$. However, when using these invariants to study the equivalence of MSRD codes in $\mathrm{Mat}(\bfm, \bfm, \F_q)$ with $t \leq q - 1$, this assumption is not restrictive. Indeed, by \Cref{prop:numberweightMSRD}, both $\C$ and $\C'$ contain a codeword with sum-rank weight, equal to $tm$. Therefore, up to equivalence, we may assume that both codes contain the element $(I_{m_1}, \ldots, I_{m_t})$.
\end{remark}

Since the left and right idealisers, the centraliser, and the center of a rank-metric code are extensions of the nuclei of a semifield, in \cite{sheekey2020new} the notion of nuclear parameters for rank-metric codes has been introduced. These parameters include the sizes of the code, its left and right idealisers, its centraliser, and its center, and they serve as invariants for rank-metric codes.

In the same spirit, and in light of \Cref{prop:equivcodeimpliesequivideal} and \Cref{prop:equivalencecenter}, we adopt analogous terminology for sum-rank metric codes.  Precisely, by analogy with the nuclear parameters of semifields and rank-metric codes, we refer to the sizes of the left and right idealisers, the centraliser, and the center as the \emph{nuclear parameters} of a sum-rank metric code, as these quantities are invariants for a sum-rank metric code.

\begin{definition}
Let $\C$ be a sum-rank metric code in $\Mat(\bfm,\bfm,\F_q)$ that contains $(I_{m_1},\ldots,I_{m_t})$.
The \textbf{nuclear parameters} of $\C$ is the tuple
\[
(|\C|,|\lid(\C)|,|\lid(\C)|,|C(\C)|,|Z(\C)|).
\]
\end{definition}
If $\C$ is a sum-rank metric code in $\Mat(\bfm,\bfm,\F_q)$ that contains a maximum weight codeword, we define its nuclear parameters as the nuclear parameters of any code $\C'$ that contains the tuple $(I_{m_1},\ldots,I_{m_t})$ and is equivalent to $\C$.
In the next section, we will use the nuclear parameters to study the equivalence among the families of linearized Reed-Solomon codes, additive twisted linearized Reed-Solomon codes and twisted linearized Reed-Solomon of TZ-type.

\subsection{Linearity of a sum-rank metric code}\label{sec:linearity}

The concept of subspace linearity cannot serve as an invariant for sum-rank metric codes.
Indeed, in the following example we show two equivalent rank-metric codes with different linearity. To this aim, we will use the framework of linearized polynomials as it is easier to check the linearity of the codes.

\begin{example} \label{ex:notlinear}
    Let $\theta \colon a \in \mathbb{F}_{2^4}\mapsto a^2 \in \mathbb{F}_{2^4}$ and consider $\mathcal{C}_1=\langle \mathrm{id}\rangle_{\F_{2^4}} \subseteq \mathbb{F}_{2^4}[\theta]$. 
    Let $p_1=\theta+\xi \theta^3 \in \mathbb{F}_{2^4}[\theta]$, where $\xi \in \F_{2^4}$ with $\mathrm{N}_{2^4/2}(\xi)\ne 1$ (so that $p_1$ is invertible), and let $p_2=\theta$.
    The code 
    \[ \C_2=p_1 \circ \C_1 \circ p_2=\{ \alpha \mathrm{id}+\alpha^2\theta+\alpha^4\theta^2\colon \alpha \in \F_{2^4} \} \subseteq \F_{2^4}[\theta] \]
    is equivalent to $\C_1$, by definition.
    The code $\C_1$ is $\F_{2^4}$-linear but it is easy to see that $\C_2$ is strictly $\F_2$-linear, which means that $\C_2$ is not $\F_{2^i}$-linear for any $i\geq 2$.
\end{example}

From the above example it is clear that the linearity of a code is not an invariant for code equivalence. Although $\C_2$ is not $\F_{2^4}$-linear, by Proposition \ref{prop:equivcodeimpliesequivideal}, its left idealiser still keeps track of the linearity of $\C_1$ as it contains a subring isomorphic to $\F_{2^4}$.
Clearly, if $\C$ is an $\F_{q^m}$-subspace of $\F_{q^m}[\theta]$ then its left idealiser contains a subring isomorphic to the field $\F_{q^m}$.
Based on the previous considerations, it is natural to give the following definition. 

\begin{definition}\label{def:linearity}
    Let $\C$ be an additive sum-rank metric code in $\spacmn$.
    If $(\lid(\C),+,\cdot)$ contains a subring $\mathcal{F}$ isomorphic to $\mathbb{F}_{q^i}$, where $+$ and $\cdot$ are respectively the sum and the product of matrices, then we say that $\C$ is $\F_{q^i}$-\textbf{linear}. 
\end{definition}

The above definition is also motivated by the fact that $\C$ turns out to be an $\mathcal{F}$-left vector space.
Proposition \ref{prop:equivcodeimpliesequivideal} implies the following result.

\begin{proposition}
    Let $\C_1$ and $\C_2$ be two sum-rank metric codes in $\spacmn$ and suppose that $\C_1$ and $\C_2$ are equivalent. If $\C_1$ is $\F_{q^i}$-linear then $\C_2$ is $\F_{q^i}$-linear.
\end{proposition}

We are going to prove that  $\fqm$-linear sum-rank metric codes in $\Mat(\bfm,\mathbf{n},\fq)$, with $m=m_1=\ldots=m_t$ can be considered in the vector framework $\fqm^{\mathbf{n}}$.
We start by recalling the framework $\fqm^{\mathbf{n}}$. Sum-rank metric codes have introduced in the framework of matrices, but in some cases it is very useful to look at sum-rank metric codes in 
\[\F_{q^m}^{\bfn}=\bigoplus_{i=1}^t \F_{q^m}^{n_i}.\]

Let start by recalling that the rank of a vector $v=(v_1,\ldots,v_n) \in \F_{q^m}^n$ is the dimension of the
vector space generated over $\F_q$ by its entries, i.e, $\rk(v)=\dim_{\fq} \langle v_1,\ldots, v_n\rangle_{\fq}$. For an element $\mathbf{x}=(x_1 , \ldots,  x_t)\in \F_{q^m}^{\bfn}$, with $x_i\in\F_{q^m}^{n_i}$, the \textbf{$\Fq$-sum-rank weight of $\xx$} is defined as
$$ \w_{\srk}(\xx)=\sum_{i=1}^t \rk(x_i).$$
A \textbf{(vector) sum-rank metric code} $\C$ is a subset of $\fqm^{\bfn}$, endowed with the 
\textbf{$\Fq$-sum-rank distance} on $\fqm^{\bfn}$, defined as 
\[
\dsrk(\xx,\yy)=\w_{\srk}(\xx-\yy),
\]
for any  $\xx, \yy \in \fqm^{\bfn}$. The minimum (sum-rank) distance of a code $\C$ is denoted $d(\C)$, that is,
\[
d(\C):=\min\{\dsrk(\xx,\yy) \st \xx, \yy \in \C, \xx\neq \yy  \}.
\]
A vector sum-rank metric code $\mathcal{C} \subseteq \F_{q^m}^{\bfn}$ is said \textbf{linear} if it is an $\F_{q^m}$-subspace of $\F_{q^m}^{\bfn}$. We write that $\C$ is a linear $\Fmnkd$ code, if $\C$ is a sum-rank metric code in $\fqm^{\bfn}$ having $\Fm$-dimension $k$ and minimum (sum-rank) distance $d$, or simply an $\Fmnk$ code if the minimum distance is not relevant/known.
For more details on the vector framework and its connection with the vector setting, see e.g. \cite{Martinez2018skew,neri2021twisted,ott2021bounds,puchinger2022generic}. \\

The matrix and vector settings are related as following. First let $\mathcal{B}=(\gamma_1,\ldots,\gamma_m)$ be an $\F_q$-basis of $\fqm$. For a vector $x \in \F_{q^m}^n$, define $\Ext_{\mathcal{B}}(x) \in \F_q^{m\times n}$ the matrix expansion of the vector $x$ with respect to the $\fq$-basis $\mathcal{B}$ of $\F_{q^m}$, i.e., if $x=(x_1,\ldots,x_n)$, then
$$x_i = \sum_{j=1}^m (\Ext_{\mathcal{B}}(x))_{ij}\gamma_i, \qquad \mbox{ for each } i \in \{1,\ldots,n\}.$$

Let $\mathcal{B}_r=(\gamma_1^{(r)},\ldots,\gamma_m^{(r)})$ be an ordered $\fq$-basis of $\F_{q^m}$, for each $r \in \{1,\ldots,t\}$ and define $\boldsymbol
{\mathcal{B}}=(\mathcal{B}_1,\ldots,\mathcal{B}_t)$. For an element $\xx=(x_1, \ldots ,x_t) \in \fqm^{\bfn}$, with $x_i \in \F_{q^m}^{n_i}$, we define     \begin{equation} \label{eq:isometrymatrixvector}
\Ext_{\boldsymbol
{\mathcal{B}}}(\xx)=(\Ext_{\mathcal{B}_1}(x_1), \ldots, \Ext_{\mathcal{B}_t}(x_t)) \in \Mat(\bfm,\bfn,\F_q),\end{equation}
with $\bfm=(m,\ldots,m)$.
It is immediate to verify that the map
$\Ext_{\boldsymbol
{\mathcal{B}}}: \bigoplus_{i=1}^t\Fm^{n_i} \longrightarrow \Mat(\bfm,\bfn,\F_q)$
is an $\fq$-linear isometry between the metric spaces $(\fqm^{\bfn}, \dsrk)$ and $(\Mat(\bfm,\bfn,\F_q),\dsrk)$.

An important property of the map $\Ext_{\boldsymbol
{\mathcal{B}}}$, already proved in \cite{polverino2022divisible} for the rank metric, is that it keeps track of the $\F_{q^m}$-linearity in the following sense. 

\begin{proposition}
If $\C$ is a nondegenerate $[\bfn,k]_{q^m/q}$-code then $\C'=\Ext_{\boldsymbol
{\mathcal{B}}}(\C)$ is an $\F_{q^m}$-linear sum-rank metric code in $\Mat(\bfm,\bfn,\F_q)$, with $\bfm=(m,\ldots,m)$. Conversely, let $\C$ be a nondegenerate $\F_{q^m}$-linear sum-rank metric code in $\Mat(\bfm,\bfn,\F_q)$, with $\bfm=(m,\ldots,m)$, of dimension $mk$. Then there exists a code $\C'$ equivalent to $\C$ such that $\Ext_{\boldsymbol
{\mathcal{B}}}^{-1}(\C')$ is an $[\bfn,k]_{q^m/q}$ code.
\end{proposition}
\begin{proof}
For any element $\alpha \in \F_{q^m}$, define
\[
\begin{array}{rrcl}
     \tau_{\alpha} : & \F_{q^m} & \longrightarrow & \F_{q^m}  \\
     & x & \longmapsto & \alpha x. 
\end{array}
\]
The set $\mathcal{A}_i$ of the matrices in $\F_q^{m \times m}$ associated with $\tau_{\alpha}$ with respect to the basis $\mathcal{B}_i$, for every $\alpha \in \F_{q^m}$, forms a ring isomorphic to $\F_{q^m}$.
Hence,
\begin{equation}\label{eq:Gammaqlin}
  \Ext_{\boldsymbol
{\mathcal{B}}}(\alpha c) = (A_1, \ldots , A_t) \Ext_{\boldsymbol
{\mathcal{B}}}(c),  
\end{equation}
for all $\alpha \in \F_{q^m}$ and $c \in \C$, where $A_i$ is matrix in $\F_q^{m \times m}$ associated with $\tau_{\alpha}$ with respect to the basis $\mathcal{B}_i$. Now, since $\alpha c \in \C$ for all $c \in \C$, we get that $(A_1, \ldots, A_t) \in \mathcal{I}_\ell(\C)$ and the first part of the assertion is proved.
Conversely, let $\C$ be an $\F_{q^m}$-linear sum-rank metric code in $\Mat(\bfm,\bfn,\F_q)$, with $\bfm=(m,\ldots,m)$ and let $\mathcal{G}$ be a subring of $\mathcal{I}_{\ell}(\C)$ isomorphic to $\F_{q^m}$. 
Denote by $\mathcal{G}_i=\iota_i(\mathcal{G})$ for any $i \in \{1,\ldots,t\}$, and observe that $\mathcal{A}_i \setminus  \{0\}$ and $\mathcal{G}_i \setminus  \{0\}$ define two Singer cycles in $(\GL(m,q),\cdot)$, and since two Singer cycles are conjugate (see \cite[pag. 187]{huppert2013endliche} and \cite[Section 1.2.5 and Example 1.12]{hiss2011finite}), then there exists an invertible matrix $H_i \in \mathrm{GL}(m,q)$ such that 
\[H_i^{-1}  \mathcal{G}_i H_i=\mathcal{A}_i,\]
for any $i \in \{1,\ldots,t\}$.
Consider $\C'=(H_1^{-1},\ldots,H_t^{-1})\circ_{\mathbf{1}}\C \circ_{\bf{1}} (H_1,\ldots,H_t)$. Now, it is easy to verify that $\Ext_{\boldsymbol
{\mathcal{B}}}^{-1}(\C')$ is an $[\bfn,k]_{q^m/q}$ code (as each $\iota_i(\C)$ is an $\fqm$-subspace of $\fqm^{n_i}$ for any $i$).
\end{proof}


\section{Skew polynomial framework for sum-rank metric codes}\label{sec:skewpolSRC}

The skew polynomial framework has been fundamental for both rank-metric and sum-rank metric codes in constructing families of MSRD codes, see e.g. \cite{sheekey2020new,neri2021twisted,lobillo2025quotients,thompson2023division}. For the sum-rank metric, this framework is notably compact. Specifically, a tuple of matrices can be identified with an element of a certain quotient algebra of a skew polynomial ring. In this context, we can directly describe the largest known families of MSRD codes:
\begin{itemize}
\item Linearized Reed-Solomon (LRS) codes,
\item Additive Twisted Linearized Reed-Solomon (ATLRS) codes,
\item Twisted Linearized Reed-Solomon (TLRS) codes of TZ-type.
\end{itemize}
In this section, after reviewing the skew polynomial framework and the mentioned codes, we will calculate the nuclear parameters of these families of codes. These parameters will aid in studying the equivalence issues among these MSRD code families.


\subsection{Skew-polynomial setting}

Skew polynomial rings have been widely used in literature in connection with different topics, starting with the pioneering work of Ore \cite{ore1933theory}. For further background and properties of these rings, we refer to \cite{jacobson2009finite,goodearl2004introduction,ore1933theory}. Let $\theta$ be a generator of the Galois group $\Gal(\F_{q^m}/\F_q)$.
The ring of skew polynomials $\F_{q^m}[x;\theta]$ is the set of formal polynomials
\[
\F_{q^m}[x;\theta]:=\left\{f(x)=\sum_{i=0}^d f_i x^i \st f_i \in \F_{q^m}, d \in \Z_{\geq 0} \right\}
\]
equipped with the ordinary polynomial addition and the multiplication rule 
\[
x \cdot a=\theta(a) \cdot x, \mbox{ for every } a \in \F_{q^m},
\]
extended to polynomials by associativity and distributivity. In the next, we let $R:=\F_{q^m}[x;\theta]$. The ring $R$ is a noncommutative ring whose center is $\F_q[x^m]$.  The ring $R$ admits a natural notion of \emph{degree} for nonzero elements, denoted $\deg(f)$ for $f \in R$, which behaves analogously to the classical degree of commutative polynomials.

Moreover, $R$ is a \textbf{left Euclidean domain}: for any pair of nonzero elements $f, g \in R$, there exist unique $q, r \in R$ such that
\[
f = qg + r,
\]
where either $r = 0$ or $\deg(r) < \deg(f)$. When $r = 0$, we say that $g$ \textbf{right-divides} $f$, denoted by $g \mid_r f$, and that $f$ is a \textbf{left-multiple} of $g$. As a consequence, $R$ turns out to be a left (and right) \textbf{principal ideal domain}. Also, for any two nonzero elements $f_1, f_2 \in R$, one can define their \textbf{greatest common right divisor}, denoted $\mathrm{gcrd}(f_1, f_2)$, and their \textbf{least common left multiple}, denoted $\mathrm{lclm}(f_1, f_2)$. These are characterized, respectively, as generators of the sum and intersection of the left ideals they generate: if $a = \gcrd(f_1, f_2)$ and $b = \mathrm{lclm}(f_1, f_2)$, then
\[
Rf_1 + Rf_2 = Ra \quad \text{and} \quad Rf_1 \cap Rf_2 = Rb,
\]
where $Rf_i$ denotes the left ideal in $R$ generated by $f_i$.

In \cite{Martinez2018skew}, for the first time sum-rank metric codes were studied in the context of skew polynomials. Later, in \cite{neri2021twisted}, a skew polynomial framework was introduced for sum-rank metric codes. This framework is isometric to the natural setting of tuple of matrices for sum-rank metric codes. In the following, we recall this framework and work within it.

For the rest of the section, we fix the following notation. Let $\alpha_1,\ldots,\alpha_t \in \F_{q^m}^*$ be such that $\N_{q^m/q}(\alpha_i) \neq \N_{q^m/q}(\alpha_j)$, for any $i,j \in \{1,\ldots,t\}$ with $i \neq j$. Consider $\Lambda=\{\lambda_1,\ldots,\lambda_t\}$ where $\lambda_i=\N_{q^m/q}(\alpha_i)$ for any $i$. Define
\[
H_{\Lambda}(x)=\prod_{i=1}^t(x^m-\lambda_i) \in R.
\]  

Note that $H_{\Lambda}$ is the product of some central elements of $R$ and hence it belongs to the center of $R$. Also, $H_{\Lambda}$ generates a two-sided ideal in $R$.
Therefore, we can consider the quotient ring \( R / R \Hl \).

For an element $f(x) = f_0+f_1x+\cdots+f_d x^d \in R$ and $\alpha \in \F_{q^m}^*$, denote by $f_{\alpha}$ the skew polynomial
\[
f_{\alpha}(x):=\sum_{i=0}^d f_i \N_{\theta}^i(\alpha)x^i,
\]
where \( \mathrm{N}_{\theta}^i(\alpha) = \prod_{j=0}^{i-1} \theta^j(\alpha) \) denotes the \( i \)-th truncated norm of \( \alpha \), for all integers \( i \geq 0 \).

Skew polynomial rings are related to the skew group algebra $\F_{q^m}[\theta]$, and in fact $\F_{q^m}[\theta]$ can be obtained as specific quotients of the skew polynomial rings.  Precisely, linearized polynomials and skew polynomials are related by the map
\begin{equation} \label{eq:skewandlinearized}
\begin{tabular}{l c c c }
$\Phi:$ & $R$ & $\longrightarrow$ & $\F_{q^m}[\theta]$ \\
& $f_0+f_1x+\cdots+f_dx^d$ & $\longmapsto$ & $f_0\mathrm{id}+f_1 \theta +\cdots+f_d \theta^d$.
\end{tabular}
\end{equation}

The map $\Phi$ is an $\F_q$-algebra epimorphism whose kernel is the two-sided ideal $R(x^m-1)$. Hence, $\Phi$ induces an $\F_q$-algebra isomorphism 
\[
\begin{tabular}{l c c c }
$\overline{\Phi}:$ & $\frac{R}{R(x^m-1)}$ & $\longrightarrow$ & $\F_{q^m}[\theta]$ \\
& $f_0+f_1x+\cdots+f_{m-1}x^{m-1}+ R(x^m-1)$ & $\longmapsto$ & $f_0\mathrm{id}+f_1 \theta +\cdots+f_{m-1} \theta^{m-1}$.
\end{tabular}
\]

Now, we can give a skew polynomial description of the space $(\F_{q^m}[\theta])^t$, extending the map $\Phi$ defined in \eqref{eq:skewandlinearized}, which will allow us to deal with sum-rank metric codes in the framework of skew polynomial rings.

\begin{theorem} [see \textnormal{\cite[Theorem 8]{martinez2022theory}}] \label{th:multiplelinearizedskew} 
The map 
\[
\begin{tabular}{l c c c }
$\Phi_{\boldsymbol{\alpha}}:$ & $R$ & $\longrightarrow$ & $(\F_{q^m}[\theta])^t$ \\
& $f$ & $\longmapsto$ & $(\Phi(f_{\alpha_1}), \ldots,\Phi(f_{\alpha_t}) )$
\end{tabular}
\]
is an $\F_q$-algebra epimorphism, whose kernel is the two-sided ideal $RH_{\Lambda}$. Hence, $\Phi_{\boldsymbol{\alpha}}$ induces an $\F_q$-algebra isomorphism
\[
\begin{tabular}{l c c c }
$\overline{\Phi}_{\boldsymbol{\alpha}}:$ & $\frac{R}{R\Hl}$ & $\longrightarrow$ & $(\F_{q^m}[\theta])^t$ \\
& $f+R\Hl$ & $\longmapsto$ & $(\Phi(f_{\alpha_1}), \ldots,\Phi(f_{\alpha_t}) )$.
\end{tabular}
\]
\end{theorem}

In the following, when we write \( a \in R/R\Hl \), we usually denote an element of this form
\[
a = \sum_{i=0}^{nt-1} a_i x^i + R\Hl,
\]
i.e., the equivalence class \( a \) is identified with the unique skew polynomial \( a(x) = \sum_{i=0}^{nt-1} a_i x^i \in R \) of minimal degree representing it.

The following notion of sum-rank weight on the quotient space $R/RH_{\Lambda}$ has been defined. 

\begin{definition}
   The \textbf{sum-rank weight} on the space $R/RH_{\Lambda}$ of an element $a$ is  
    \[
    \wl(a)=tm-\deg(\gcrd(a,\Hl)).
    \]
    Moreover, the sum-rank weight induces the \textbf{sum-rank distance} on $\RH$, which is defined as
\[
\dl(a,b):=\wl(a-b),
\]
for every $a,b \in \RH$.
\end{definition}

So, in the metric space $(R/R\Hl,\dl)$ an \textbf{(additive) sum-rank metric code} is an $\F_p$-subspace of $R/R\Hl$.
The definition of sum-rank weight and distance on $R/RH_{\Lambda}$ is consistent with the metric $\dsrk$ on the space $(\F_{q^m}[\theta])^t$. In other words,
\[
\dl(f,g)=\dsrk(\overline{\Phi}_{\boldsymbol{\alpha}}(f),\overline{\Phi}_{\boldsymbol{\alpha}}(g)),
\]
for every $f,g \in \RH$, where $\dsrk(\cdot,\cdot)$ is defined as in \eqref{eq:sumrankpoly}, cf. \cite[Corollary 4.5]{neri2021twisted}. As a consequence, the map 
\begin{equation} \label{eq:isometryskewskew}
\overline{\Phi}_{\boldsymbol{\alpha}}:\left(\frac{R}{R\Hl},\dl\right) \rightarrow ((\F_{q^m}[\theta])^t,\dsrk)
\end{equation}
is an isometry of metric spaces.

\begin{remark} \label{rk:changealpha}
    The isometry $\overline{\Phi}_{\boldsymbol{\alpha}}$, clearly depends on the chosen 
elements $\alpha_1,\ldots,\alpha_t$, and, consequently, on the set $\Lambda$ of their norms. However, it is worth noting that any set of elements $\alpha_1,\ldots,\alpha_t$ with pairwise distinct norms also induces an isometry as well.
\end{remark}

Thanks to the isometry between the skew group algebra and the skew polynomial frameworks, one can connect the frameworks of skew polynomial and matrix settings. Specifically, using the isometry $\overline{\phi}$ given in \eqref{eq:ismotrymatrpoly}, the map
\begin{equation} \label{eq:algebramatrixskew}
\begin{tabular}{l c c c }
$\overline{\phi} \circ \overline{\Phi}_{\boldsymbol{\alpha}}:$ & $\frac{R}{R\Hl}$ & $\longrightarrow$ & $\Mat(\bfm,\bfm,\F_q)$ \\
& $f$ & $\longmapsto$ & $(\phi_{\Phi(f_{\alpha_1})}, \ldots,\phi_{\Phi(f_{\alpha_1})}) $
\end{tabular}
\end{equation}
where $m = m_1 = \ldots = m_t$, is an $\F_q$-algebra isomorphism that also defines an isometry between the metric spaces
\begin{equation} \label{eq:isometryskewmatrix}
\overline{\phi} \circ \overline{\Phi}_{\boldsymbol{\alpha}}:\left(\frac{R}{R\Hl},\dl\right) \rightarrow (\mathrm{Mat}(\bfm,\bfm,\F_q),\dsrk).
\end{equation}
We can therefore transfer all the notions developed in the matrix-based setting to the skew polynomial framework, while preserving the metric properties of the corresponding objects.


As proved in \cite[Sections 4 and 5]{neri2021twisted}, it is possible describe the dual of a code $\C$ within the skew polynomial framework consistently with respect to the matrix setting. We will recall it in the next. Assume that $\Lambda$ is a  (cyclic) subgroup of $\F_q^*$.

Consider the following nondegenerate $\F_p$-bilinear form $\langle \cdot , \cdot \rangle_{\Lambda}$ on $\RH$ defined as
\[
\langle f , g \rangle_{\Lambda}:= \mathrm{Tr}_{q^m/p}\left( \sum_{i=0}^{tm-1}f_ig_i \right),
\]
for every $f=\sum_{i=0}^{tm-1}f_ix^i+R\Hl,g=\sum_{i=0}^{tm-1}g_ix^i+R\Hl \in \RH$. 

We recall that it is possible to choose an \(\mathbb{F}_q\)-basis for \(\mathbb{F}_{q^m}\) so that, under the isometry described in \eqref{eq:isometryskewmatrix}, this bilinear form coincides with the one defined in \eqref{eq:bilinearmatrix}; see \cite[Section 5.1]{neri2021twisted}.
 
\begin{definition}
    Let $\C \subseteq \RH$. The \textbf{dual code} of $\C$ is 
    \[
    \C^{\perp}:=\left\{f \in \frac{R}{R\Hl} \colon \langle f , g \rangle_{\Lambda}=0 \mbox{ for every }g \in \C\right\}.
    \]
\end{definition}

\subsection{Nuclear parameters of MSRD codes}

The skew polynomial framework offers a key advantage: it enables the translation of concepts from sum-rank metric codes - originally formulated in the matrix or skew algebra settings - into the language of skew polynomials. This allows us to represent a tuple of matrices or a tuple of linearized polynomials as a single skew polynomial within a suitable quotient ring. 

We begin by reformulating the notions of idealisers, centralizers, and the center of a sum-rank metric code in terms of skew polynomials.

\begin{proposition} \label{prop:invariantsskew}
    Let $\C \subseteq \Mat(\bfm,\bfm,\F_q)$, with $m=m_1=\ldots=m_t$, and $t \leq q-1$. Then the following ring isomorphisms hold
 \[
    \begin{array}{rl}
       \lid(\C) \simeq   & \left\{g \in \RH \colon gf \in \C, \mbox{ for every }f \in \C\right\}, \\
         \rid(\C) \simeq   & \{g \in \RH \colon fg \in \C, \mbox{ for every }f \in \C\}, \\
         \mathrm{C}(\C) \simeq  & \{g \in \RH \colon gf=fg , \mbox{ for every }f \in \C\},\\
         Z(\C) \simeq &  \{g \in \RH \colon gf_1 \in \C \mbox{ and }gf_2=f_2g \mbox{ for every }f_1,f_2 \in \C\} 
    \end{array}
    \]
\end{proposition}

\begin{proof}
The statement follows directly by using the ring isomorphism between $\Mat(\bfm, \bfm, \F_q)$ and $\RH$ given in \eqref{eq:algebramatrixskew}.
\end{proof}

Therefore, left and right idealisers, centralisers and the centre can be used as invariants for codes defined in the skew polynomial framework $\RH$.

\begin{remark}
Note that the isometry defined in \eqref{eq:isometryskewmatrix} satisfies
\[
\overline{\phi} \circ \overline{\Phi}_{\boldsymbol{\alpha}}(1) = (\mathrm{id}, \ldots, \mathrm{id}).
\]
Therefore, if \(\mathcal{C}\) is a sum-rank metric code in \(R / R \Hl\), the condition \(1 \in \mathcal{C}\) corresponds to the condition \((I_m, \ldots, I_m) \in \overline{\phi} \circ \overline{\Phi}_{\boldsymbol{\alpha}}(\mathcal{C}) \subseteq \operatorname{Mat}(\mathbf{m}, \mathbf{m}, \mathbb{F}_q)\). Therefore, in the following, when computing the centraliser of a code \(\mathcal{C}\), we assume that \(1 \in \mathcal{C}\).
\end{remark}

The first construction of MSRD codes was introduced in \cite{Martinez2018skew} by Martínez-Peñas, and such codes are refereed as \emph{linearized Reed-Solomon codes}. These are the sum-rank metric analog of Gabidulin codes in the rank metric and Reed-Solomon codes in the Hamming metric. In \cite{neri2021twisted}, Neri introduced a new family of MSRD codes. These codes are known as \emph{additive twisted linearized Reed-Solomon codes}, as they can be considered an extension in the sum-rank metric of additive twisted Gabidulin codes \cite{sheekey2016new,otal2016additive} in the rank metric and twisted Reed-Solomon codes in the Hamming metric \cite{beelen2018structural,beelen2022twisted}. 

\begin{definition}[see \textnormal{ \cite[Definition 31]{Martinez2018skew} and \cite[Definition 6.2]{neri2021twisted}}]\label{def:ATLRScodes}
Let $\tau \in \Aut(\F_{q^m})$. For every $1 \leq k < tm$, the code
\[
\C_{k,\theta}(\eta,\tau)=\left\{ f_0+\ldots+f_{k-1}x^{k-1}+\eta \tau(f_0)x^k+R\Hl \colon f_i \in \F_{q^m}\right\} \subseteq \frac{R}{R\Hl}, 
\]
is called \textbf{additive twisted linearized
Reed-Solomon (ATLRS) code}.
\end{definition}

\begin{theorem} [see \textnormal{\cite[Theorem 4]{Martinez2018skew}} and \textnormal{\cite[Theorem 6.3 and Remark 6.7]{neri2021twisted}}]
    Consider $\F:=\fq \cap \Fix(\F_{q^m},\tau)$ and let $u=[\fq:\F]$. Let $\eta \in \F_{q^m}$ such that \[(-1)^{ukm}\N_{\F_{q^m}/\F}(\eta) \notin \langle \Lambda \rangle,\] where $\langle \Lambda \rangle$ denotes the multiplicative subgroup of $\F_{q}^*$
generated by $\Lambda$. For every $1 \leq k <tm$, the code $\mathcal{C}_{k,\theta}(\eta,\tau)$ is an MSRD code
\end{theorem}

When $\eta = 0$, these codes coincide with the linearized Reed–Solomon codes, and we define \[
\C_{k,\theta} := \C_{k,\theta}(0, \mathrm{id}).
\] 

The main problem regards to show whether the family of ATLRS codes contains examples that are not equivalent to linearized Reed-Solomon codes.
In order to solve this problem, we will study the nuclear parameters of the codes of these two families. 

We begin by computing the left and right idealisers of \(\mathcal{C}_{k, \theta}(\eta, \tau)\), assuming that \(3 \leq k \leq tm/2\). This restriction on \(k\) is necessary for our computation. Then, using a duality argument, we obtain the idealisers for the complementary range \(tm/2 + 1 \leq k \leq tm - 3\).

\begin{theorem} \label{lm:idealizersATLRS}
Let $\C=\C_{k,\theta}(\eta,\tau)$ as in \Cref{def:ATLRScodes}, with $3 \leq k \leq tm/2$. Then: 
    \begin{itemize}
    \item if $\eta=0$, we have
        \[
        \lid(\C)\simeq\F_{q^m} \ \mbox{ and }\ \rid(\C)\simeq\F_{q^m};
        \]
        \item if $\eta \neq 0$, we have
        \[
        \lid(\C)\simeq \Fix(\F_{q^m},\tau) \ \mbox{ and }\ \rid(\C)\simeq \Fix(\F_{q^m},\tau^{-1} \circ \theta^{k}).
        \]
    \end{itemize}
\end{theorem}

\begin{proof}
We start by computing $\lid(\C)$. By \Cref{prop:invariantsskew}, we know that $\lid(\C)$ consists of all the elements $g \in \RH$ such that $g \C \subseteq \C$. First, we prove that any element of $\lid(\C)$ has degree at most $k-1$. If $\eta=0$, then $1 \in \C$ and so $\lid(\C) \subseteq \C$, implying that any $g \in \lid(\C)$ has degree at most $k-1$. Assume now that $\eta \neq 0$, our aim is to show that if $g\in \lid(\C)$ then its degree is at most $k-1$. Let $g\in \lid(\C)$, we must have that $g\alpha x^i \in\C,$ for all $\alpha \in \F_{q^m}$ and $i\in \{1,\ldots,k-1\}$. As $k\geq 3$, this set is non-empty. 
If $g=\sum_{i=0}^{tm-1}g_ix^i+R\Hl$ and $\Hl(x)=H_0+\ldots+H_{t-1}x^{(t-1)m}+x^{tm}$, we get
\[
gx = \left(\sum_{i=1}^{mt-1} g_{i-1}x^i\right) - g_{mt-1}\left( \sum_{i=0}^{t-1}H_i x^{mi} \right) +R\Hl.
\]
Hence, for all $j\in \{k+1,\ldots,mt-1\}$, we have
\begin{equation} \label{eq:g_{mt-1}}
g_{j-1} = g_{mt-1}H_{j/m},
\end{equation}
where $H_{j/m}:=0$, if  $j \not\equiv 0 \pmod m$. In particular, note that 
\begin{equation} \label{eq:g_{mt-2}}
    g_{mt-2}=0.
\end{equation}
We prove that $g_{mt-1}=0$. The coefficient of $x^{k}$ in $ gx$ is $g_{k-1}-g_{mt-1}H_{k/m}$. While the constant coefficient is $-g_{mt-1}H_0$. Therefore, as $gx \in \C$, we have
\begin{equation}\label{eq:etatauH0g}
g_{k-1} -g_{mt-1}H_{k/m}= \eta \tau(H_0g_{mt-1}).
\end{equation}
Additionally, since $k>2$, we also get that $gx^2 \in \C$. The coefficient of $x^{k+1}$ in $gx^2$ is \[g_{k-1} -g_{mt-1}H_{k/m}-g_{mt-2}H_{(k+1)/m}=g_{k-1} -g_{mt-1}H_{k/m},\]
where the equality follows by \eqref{eq:g_{mt-2}}. Since $gx^2 \in \C$, such coefficient must be zero. Hence, we have that $\eta \tau(H_0g_{mt-1})=0$ by Equation \eqref{eq:etatauH0g}. Since $\eta\ne 0 $ and $H_0\ne 0$, we must have $g_{mt-1}=0$, as claimed. As a consequence, by \eqref{eq:g_{mt-1}}, we get that $\deg(g)\leq k-1$ also in the case in which $\eta \ne 0$.
Therefore, for any $\eta$ we have 
\[\lid(\C) \subseteq \{g \in \RH  \colon \deg(g)\leq k-1\}.\] So, let $g\in \lid(\C)$, and $\deg(g)\leq k-1$. By using that $k>1$, we have $x^{k-1}\in \C$, and hence $gx^{k-1} \in \C$. Now since $\deg(gx^{k-1})\leq 2k-2<mt$ and 
\[gx^{k-1} = g_0x^{k-1}+g_1x^{k}+\ldots g_{sk}x^{2k-2} +R\Hl,\]we get that $g$ is constant and $g = g_0+R\Hl$. Now, since $f_0+\eta \tau(f_0)x^k +R\Hl\in \C$, it follows that $g_0(f_0+\eta \tau(f_0)x^k)+R\Hl\in \C$, or equivalently, $g_0\eta \tau(f_0) = \eta \tau (g_0f_0)$. Thus, if $\eta \ne 0$ then $\tau(g_0)=g_0$. Hence 
\[
\lid(\C) \simeq \begin{cases}\Fix(\F_{q^m},\tau) &\eta \ne 0\\ \F_{q^m}&\eta=0,\end{cases}
\]
as claimed. 
Arguing as in the first part, we can show that if $g\in \rid(\C)$ then $\deg(g)\leq k-1$, and $g=g_0+R\Hl$. Now, $(a_0+\eta \tau(a_0)x^k)g_0+R\Hl\in \C$ if and only if $\theta^{k}(g_0)\eta \tau(f_0)= \eta \tau(g_0f_0)$, and so if $\eta\ne 0$ then $\tau(g_0)=\theta^{k}(g_0)$. Hence 
\[
\rid(\C) \simeq \left\{\begin{array}{ll}\Fix(\F_{q^m},\tau^{-1} \circ \theta^k) &\eta \ne 0\\ \F_{q^m}&\eta=0.\end{array}\right.
\]
Hence the left idealiser and the right idealiser are as claimed.
\end{proof}

To determine the idealisers the case $k>tm/2$, we need to recall the duals of the LRS codes and of the ATLRS codes.

\begin{theorem} [see \textnormal{\cite[Theorem 5.9 and Theorem 6.11]{neri2021twisted}}] \label{th:dualreed}
Let $\Lambda$ be a cyclic subgroup of $\F_q^*$ and let $1\leq k <tm$ be a positive integer. Then
\begin{enumerate}
    \item $(\C_{k,\theta})^{\perp}=\C_{tm-k,\theta} \cdot x^k=x^k \cdot \C_{tm-k,\theta}.
$
\item $
(\C_{k,\theta}(\eta,\tau))^{\perp}=\C_{tm-k,\theta}(\tau^{-1}(\eta),\tau^{-1}) \cdot x^k.$
\end{enumerate}
\end{theorem}

Now, we are able to completely determine the nuclear parameters of LRS codes and of the ATLRS codes.

\begin{theorem}
  \label{th:idealizersany}
Let $\Lambda$ be a cyclic subgroup of $\F_q^*$. Let $\C=\C_{k,\theta}(\eta,\tau)$ as in \Cref{def:ATLRScodes}, with $tm/2+1 \leq k \leq tm-3$. Then: 
    \begin{itemize}
          \item if $\eta=0$, we have
        \[
        \lid(\C) \simeq \F_{q^m}, \ \  \ \rid(\C)\simeq\Fix(\F_{q^m},\tau^{-1} \circ \theta^{k});
        \]
           \item if $\eta \neq 0$, we have
        \[
\lid(\C)\simeq\Fix(\F_{q^m},\tau) \ \mbox{ and }  \ \rid(\C)\simeq\Fix(\F_{q^m},\tau^{-1} \circ \theta^{k}).
        \]
    \end{itemize}
\end{theorem}

\begin{proof}
Since $tm/2+1 \leq k \leq tm-3$, we have
$3 \leq tm-k \leq tm/2-1$, and by \Cref{th:dualreed}, we have that $\C^{\perp}$ is equivalent to $\mathcal{C}_{tm-k,\theta}(\tau^{-1}(\eta),\tau^{-1})$. Therefore, by \Cref{prop:dualtranspideal}, we obtain
\[
\lid(\C) \simeq \lid(\C^{\perp})  \simeq \lid(\mathcal{C}_{tm-k,\theta}(\tau^{-1}(\eta),\tau^{-1}))
\]
and 
\[
\rid(\C) \simeq \rid(\C^{\perp}) \simeq \rid(\mathcal{C}_{tm-k,\theta}(\tau^{-1}(\eta),\tau^{-1})).
\]
Hence, by using \Cref{lm:idealizersATLRS}, we prove that the idealisers are as claimed. 
\end{proof}

\begin{remark} \label{rk:lengthreed}
As shown in \cite[Proposition 6.8]{neri2021twisted}, the maximum possible number of blocks \(t\) for which ATLRS codes with \(\eta \neq 0\) turns out to be an MSRD code is \(t \leq \frac{q-1}{r}\), where \(r\) is the smallest prime dividing \(q-1\). This maximal length is achieved by choosing \(\Lambda\) as the largest proper subgroup of \(\mathbb{F}_q^*\). Therefore, the condition on \(\Lambda\) to be a cyclic subgroup of $\F_q^*$ allows us to construct ATLRS codes with the maximum possible number of blocks and is not overly restrictive.   
\end{remark}

In the next, we determine the other two nuclear parameters of LRS and ATLRS codes, the centralizers and the center. Since the computation of the center relies on the left idealiser, and the latter has been determined for \(tm/2 + 1 \leq k \leq tm-1 \) in the case where \(\Lambda\) is a cyclic group, in the following proposition, we will implicitly assume that \(\Lambda\) is cyclic whenever this condition for such values of $k$.

\begin{theorem}
\label{th:centralisercenterreed}
Let $\C=\C_{k,\theta}(\eta,\tau)$ as in \Cref{def:ATLRScodes}, with $2 \leq k \leq tm-1$. 
\begin{itemize}
\item If $\eta=0$, we have
    \[\mathrm{C}(\C')\simeq \F_q^t \ \mbox{ and } \ \mathrm{Z}(\C')\simeq \Fix(\F_q).
        \]
\item If $\eta\neq 0$, for every code $\C'$ containing the identity equivalent to $\C$, we have 
    \[\mathrm{C}(\C')\simeq \F_q^t \ \mbox{ and } \ \mathrm{Z}(\C')\simeq \Fix(\F_q,\tau).
        \]
        \end{itemize}
\end{theorem}

\begin{proof}
First, we find a code $\C'$ equivalent to $\C$ containing the identity. Clearly, if $\eta=0$, then $1 \in \C$ and then we can choose $\C'=\C$. Assume that $\eta \neq 0$. We note that the element $x^{m}+R\Hl \in \RH$ is such that $\wl(x^{m})=tm$. Indeed, let $h \in R$ such that $h \mid_r \Hl$ and $h \mid_r x^m$. Assume that $\Hl(x)=H_0+\cdots+H_{t-1}x^{(t-1)m}+x^{tm}$. Hence $h \mid_r x^{im}$, for every positive integer $i$. So $h \mid_r H_0$, implying that $\deg(h)=0$. Hence $\deg(\gcrd(x^m,\Hl))=0$, proving that $\wl(x^{m})=tm$. Therefore, there exists $r \in \RH$, with $\wl(r)=tm$ such that 
\[ x(x^{m-1})r=
x^{m}r + \RH =1+\RH,\]
and consequently $\wl(x^{m-1}r)=tm$, as it is the inverse of $x$ in $R/RH_{\Lambda}$.
Let 
\[
\begin{array}{rl}
  \C' :=& \C \cdot x^{m-1}r \\
  =&  \{f_1+f_2x+\cdots+f_{k-1}x^{k-2}+\eta \tau(f_0)x^{k-1}r(x)+f_0 x^{m-1}r+R\Hl\colon f_0,\ldots,f_{k-1}\in\fqm\}
\end{array}
\]
and note that $\C'$ contains the identity and it is equivalent to $\C$. \\
Now, we determine $\mathrm{C}(\C')$.
To this aim, let $g=g_0+g_1x+\cdots g_{tm-1}x^{tm-1} +R\Hl \in \RH$ such that $g \in C(\C') \setminus \{0\}$, i.e. 
    \[gf=fg,\] for every $f \in \C'$.
For any $\alpha \in \F_{q^m}$, we have $\alpha \in \C'$ and $\alpha g- g\alpha \in \RH$. Since $\deg(\alpha g- g\alpha) <mt$, we get that 
\[
\alpha \left(\sum_{i=0}^{tm-1}g_ix^i\right)=\left( \sum_{i=0}^{tm-1}g_ix^i\right)\alpha.
\]
Therefore, $\sum_{i=0}^{tm-1}g_ix^i \in \F_{q^m}[x^m]$ and \begin{equation} \label{eq:bounddegreecentral}\deg(g)<mt-m\leq mt-2.
\end{equation}
Furthermore, since $k \geq 2$ and $x+R\Hl \in \C'$, we have $ xg=g x$. Moreover, as a consequence of \eqref{eq:bounddegreecentral}, we have $\deg(xg-gx)\leq mt-1$, and so
\[
x \left(\sum_{i=0}^{tm-1}g_ix^i\right)=\left( \sum_{i=0}^{tm-1}g_ix^i\right)x.
\]
Therefore, $\sum_{i=0}^{tm-1}g_ix^i \in \F_q[x^m]=Z(R)$. This implies that \[C(\C')=\{g+R\Hl\colon g \in Z(R)\}\simeq \F_q[x^m]/(\Hl) \simeq \frac{\F_q[y]}{(y-\lambda_1)} \oplus \ldots \oplus \frac{\F_q[y]}{(y-\lambda_t)} \simeq \F_q^t.\] 
Finally, we determine the center $Z(\C')$. First, assume that \(\eta = 0\). By \Cref{lm:idealizersATLRS} and \Cref{th:dualreed}, we have \[\lid(\C') = \{\alpha + R \Hl : \alpha \in \mathbb{F}_{q^m}\}.\] Thus,
\[
\begin{array}{rl}
Z(\mathcal{C}') & = \{\alpha + R \Hl : \alpha \in \mathbb{F}_{q^m}\} \cap C(\mathcal{C}') \\
& = \{\alpha + R \Hl : \alpha \in \mathbb{F}_q\} \\
& \simeq \mathbb{F}_q.
\end{array}
\]
Similarly, if \(\eta \neq 0\), again by \Cref{lm:idealizersATLRS} and \Cref{th:dualreed}, we obtain 
\[\lid(\C')= \{\alpha + R \Hl : \alpha \in \operatorname{Fix}(\mathbb{F}_{q^m}, \tau)\}.\] Therefore, from the previous part, we obtain
\[
Z(\mathcal{C}') = \{\alpha + R \Hl : \alpha \in \operatorname{Fix}(\mathbb{F}_{q^m}, \tau) \cap \mathbb{F}_q\} \simeq \operatorname{Fix}(\mathbb{F}_q, \tau).
\]
This concludes the proof.
\end{proof}

\begin{remark}
As shown in \Cref{cor:lridealfields}, just as in the case of MRD codes in the rank metric, the left and right idealisers of MSRD codes in $\Mat(\bfm, \bfm, \F_q)$ are also fields. Moreover, it was recently proved in \cite{gomez2025adjoint} that the centraliser and the center of MRD codes are fields in the rank-metric setting. However, this property does not extend to the sum-rank metric. Indeed, as proved in \Cref{th:centralisercenterreed}, the centralisers of ATLRS and LRS codes are rings isomorphic to a direct product of $t$ copies of $\F_q$.
\end{remark}

In \cite{neri2021twisted}, also a new relevant family a MSRD codes was introduced. This construction extends in the context of the sum-rank metric the family of MRD codes provided by Trombetti and Zhou in \cite{trombetti2018new} and it can be described as follows. Let $\F_q^{(2)}$ be the multiplicative subgroup of $\F_q^*$, i.e. it is the set of all the squares over $\Fq$, i.e. 
$\F_q^{(2)}:=\{\alpha^2 \colon \alpha \in \F_q^*\}$.

\begin{definition}[see \textnormal{\cite[Definition 7.1]{neri2021twisted}}] \label{def:TZcodes}
Assume $m=2\ell \geq 2$. Let $\gamma \in \F_{q^m}$ such that $\N_{q^m/q}(\gamma) \notin \F_q^{(2)}$. Moreover, assume that $\Lambda \subseteq \F_q^{(2)}$. For every $1 \leq k < tm$, the code
\[
\mathcal{D}_{k,\theta}(\gamma)=\left\{ f_0+\sum_{i=1}^{k-1}f_ix^i+\gamma f_k x^k +R\Hl \colon f_1,\ldots,f_{k-1} \in \F_{q^m},f_0,f_k \in \F_{q^{\ell}}\right\} \subseteq \frac{R}{R\Hl}, 
\]
is called \textbf{twisted linearized
Reed-Solomon (TLRS) code of TZ-type}.
\end{definition}

In \cite[Theorem 7.2]{neri2021twisted}, it is proved that the code $\mathcal{D}_{k,\theta}(\gamma)$ is an MSRD code. Note that the condition $\N_{q^m/q}(\gamma) \notin \F_q^{(2)}$ on $\gamma$ also implies that $\gamma \notin \F_{q^{\ell}}$.

In the next, we determine the left and right idealisers, the centraliser and the center of $\mathcal{D}_{k,\theta}(\gamma)$. 
As for the ATLRS codes, we start by investigating the case $k \leq tm/2$ and by a duality argument we will obtain the case when $tm/2\leq k \leq tm-2$.

\begin{theorem} \label{lm:idealizersTZ}
Let $\C=\mathcal{D}_{k,\theta}(\gamma)$ as in \Cref{def:TZcodes}, with $2 \leq k \leq tm/2$. Then
        \[
        \lid(\C)\simeq\F_{q^{\ell}} \ \mbox{ and }\ \rid(\C)\simeq\F_{q^{\ell}}.
        \]
\end{theorem}

\begin{proof}
We start by computing $\lid(\C)$. Since $1 \in \C$, we have $\lid(\C) \subseteq \C$, implying that any $g \in \lid(\C)$ has degree at most $k$. So, let $g=\sum_{i=0}^kg_ix^i+R\Hl$. Using the fact that $k>1$, we have $x^{k-1}\in \C$, and hence $gx^{k-1} \in \C$. Now, since $\deg(gx^{k-1})\leq 2k-1<mt$ and \[gx^{k-1} = g_0x^{k-1}+g_1x^{k}+\ldots g_{k}x^{2k-1}+R\Hl,\] we get $g_2=\cdots=g_k=0$ and so $g = g_0+g_1x^k+R\Hl$. Since $\gamma x^k \in \C$, it follows that \[(g_0+g_1x^k)\gamma x^k+R\Hl=g_0x^k+g_1\theta^k(\gamma)x^{k+1}+R\Hl \in \C.\] Thus, $g_1=0$ and $g=g_0+R\Hl$. Finally, $g_0+R\Hl \in \C$ if and only if $g_0 \in \F_{q^{\ell}}$. Hence, 
\[
\lid(\C)=\{\alpha +R\Hl:\alpha \in \F_{q^{\ell}}\} \simeq \F_{q^{\ell}}
\]
as claimed. 
A similar argument can be performed to get the statement on the right idealiser.
\end{proof}

If $\Lambda$ is a cyclic subgroup of $\F_{q^m}^*$, we can determine the left and right idealisers of $\mathcal{D}_{k,\theta}(\gamma)$ even in the case where $k > tm/2$.

\begin{theorem}
  \label{th:idealizersTZany}
Let $\Lambda$ be a cyclic subgroup of $\F_q^*$. Let $\C=\mathcal{D}_{k,\theta}(\gamma)$ defined as in \Cref{def:TZcodes}, with $tm/2+1 \leq k \leq tm-2$. Then 
        \[
        \lid(\C)\simeq \F_{q^{\ell}}, \ \  \ \mbox{ and }\ \ \ \rid(\C)\simeq \F_{q^{\ell}}.
        \]
\end{theorem}

\begin{proof}
Note that $tm/2+1 \leq k \leq tm-3$, implies that
$2 \leq tm-k \leq tm/2-1$. By \cite[Theorem 7.4]{neri2021twisted}, we have \[
(\mathcal{D}_{k,\theta}(\gamma))^{\perp}=\mathcal{D}_{tm-k,\theta}(-\gamma) \cdot x^k.
\] is equivalent to $\mathcal{D}_{tm-k,\theta}(-\gamma)$. Therefore, by \Cref{prop:dualtranspideal}, we get
\[
\lid(\mathcal{D}_{k,\theta}(\gamma))= \rid(((\mathcal{D}_{k,\theta}(\gamma))^{\perp})^{\top})\simeq \rid(\mathcal{D}_{tm-k,\theta}(-\gamma))\simeq\F_{q^{\ell}}
\]
and 
\[
\rid(\mathcal{D}_{k,\theta}(\gamma)) = \lid(((\mathcal{D}_{k,\theta}(\gamma))^{\perp})^{\top})\simeq \lid(\mathcal{D}_{tm-k,\theta}(-\gamma))\simeq\F_{q^{\ell}}.
\]
As a consequence, using \Cref{lm:idealizersTZ}, we find that the idealizers of $\C$ are as claimed. 
\end{proof}

\begin{remark} Similarly to \Cref{rk:lengthreed}, we note that the maximal length \(t\) for defining a TLRS code of TZ-type is attained by choosing \(\Lambda\) as the largest proper subgroup of \(\mathbb{F}_q^*\), i.e. \(\mathbb{F}_q^{(2)}\); see \cite[Proposition 7.6]{neri2021twisted}.  
\end{remark}

\begin{theorem}
  \label{th:centraliserTZ}
Let $\C=\mathcal{D}_{k,\theta}(\gamma)$ defined as in \Cref{def:TZcodes}, with $2 \leq k \leq tm-2$ and $\ell \geq 2$. Then 
        \[
        \mathrm{C}(\C)\simeq\F_q^t \ \mbox{ and } \ \mathrm{Z}(\C)\simeq \F_q;
        \]
\end{theorem}

\begin{proof}
Note that, $1 \in \C$ and we can therefore compute $\mathrm{C}(\C)$ and $\mathrm{Z}(\C)$.
We begin by computing the centraliser $\mathrm{C}(\C)$. To this end, let $g = g_0 + g_1 x + \cdots + g_{tm-1} x^{tm-1}+R\Hl \in \RH$ be a nonzero element such that 
\[
gf = fg \quad \text{for all } f \in \C.
\]
We have $\alpha \in \C$, for every $\alpha \in \F_{q^{\ell}}$. Thus, for every $\alpha \in \F_{q^{\ell}}$ it must hold that 
\[
\alpha g - g\alpha \in R\Hl.
\]
But $\deg(\alpha g - g\alpha) < tm$, so we conclude that 
\[
\alpha \left(\sum_{i=0}^{tm-1}g_ix^i\right)=\left(\sum_{i=0}^{tm-1}g_ix^i\right) \alpha.\] Therefore, $\sum_{i=0}^{tm-1}g_ix^i \in \F_{q^m}[x^{\ell}]$ and $\deg(g) < tm - \ell \leq tm - 2$. Moreover, since $k \geq 2$ and $x+R\Hl \in \C$, we must also have $
x g = g x$.
Again, as $\deg(g) \leq tm - 2$, we get
\[
x \left(\sum_{i=0}^{tm-1}g_ix^i\right)=\left(\sum_{i=0}^{tm-1}g_ix^i\right) x,\] implying that $\sum_{i=0}^{tm-1}g_ix^i \in \F_q[x^{\ell}]$. Finally, since $\gamma x^k+R\Hl \in \C$, we must also have:
\[
(\gamma g - g \gamma)x^k=\gamma x^k g - g \gamma x^k =0 +R\Hl.
\]
Now, since \( \w_{\srk}(x^k) = tm \), it follows immediately that  \( \gamma g - g \gamma = 0 +R\Hl\), which implies that
\[
\gamma \left(\sum_{i=0}^{tm-1}g_ix^i\right)=\left(\sum_{i=0}^{tm-1}g_ix^i\right) \gamma,\] 
and so, \( \sum_{i=0}^{tm-1}g_ix^i\in \F_q[x^m] \). Hence, \[C(\C)=\{g+R\Hl\colon g \in Z(R)\}\simeq \F_q[x^m]/(\Hl) \simeq \frac{\F_q[y]}{(y-\lambda_1)} \oplus \ldots \oplus \frac{\F_q[y]}{(y-\lambda_t)} \simeq \F_q^t.\] 
We now determine the center \( Z(\C) \). By \Cref{lm:idealizersTZ} and \Cref{th:idealizersTZany}, we get\[\lid(\C) = \{\alpha + R\Hl \mid \alpha \in \F_{q^m} \}.\] Hence,
\[
\begin{array}{rl}
Z(\C')  & = \left\{ \alpha + R\Hl \colon \alpha \in \F_{q^{\ell}} \right\} \cap C(\C) \\[5pt]
& = \left\{ \alpha + R\Hl\mid \alpha \in \F_q \right\} \\[5pt]
& \simeq \F_q.
\end{array}
\]
 which concludes the proof.
\end{proof}

The following table summarizes the nuclear parameters of LRS codes, ATLRS codes, and TLRS codes of TZ-type. Recall that when \(k \geq tm/2\), the computation of the idealisers, and hence of the center, requires \(\Lambda\) to be a cyclic group.

\begin{table}[ht]
\begin{footnotesize}
\begin{center}
    \begin{tabular}{|c|c|c|} 
        \hline
         \textbf{Family} & \textbf{Nuclear parameters} & \textbf{Notes} \\ \hline
         LRS codes $\C_{k,\theta}$ & $(q^{tmk},q^m,q^m,q^t,q)$ &  \\
         (see \cite{Martinez2018skew}) &  & \\
        \hline
         ATLRS codes $\C_{k,\theta}(\eta,\tau)$ & $\left(p^{tmke},p^{(me,h)},p^{(me,ke-h)},p^{e},p^{(e,h)}\right)$ & $\tau(y)=y^{p^h}$, with $h < me$ \\ 
         (see \cite{neri2021twisted}) & & $\theta(y)=y^{p^{ej}}$, with $(j,m)=1$ \\ \hline
         TLRS codes of TZ-type & $(q^{tmk},q^{m/2},q^{m/2},q^t,q)$ & $q$ odd and $n$ even \\
         (see \cite{neri2021twisted}) &  & \\
        \hline
    \end{tabular} 
\end{center}
\caption{Nuclear parameters of LRS codes, ATLRS codes and TLRS codes of TZ-type in $\Mat(\bfm,\bfm,\F_q)$, with minimum distance $tm-k+1$. Here $q=p^e$.}
    \label{tab:parameters}
    \end{footnotesize}
\end{table}

\subsection{Comparison} 
Using the invariants that we have determined for ATLRS codes and TLRS codes of TZ-type, we can analyze their inequivalence.

Given a subset $\Lambda \subseteq \F_q^*$ and a generator $\theta$ of the Galois group $\Gal(\F_{q^m}/\F_q)$, we obtain a setting $R / R H_{\Lambda}$ that is isometric to the matrix framework $\Mat(\bfm, \mathbf{m}, \F_q)$, as previously observed. Clearly, we may also consider a different subset $\Lambda' \subseteq \F_q^*$ and a different generator $\theta'$ of $\Gal(\F_{q^m}/\F_q)$, defining a new skew polynomial ring $R' := \F_{q^m}[x; \theta']$. Nevertheless, as noted in \Cref{rk:changealpha}, the corresponding quotient ring $R' / R' H_{\Lambda'}$ is still isometric to $\Mat(\bfm,\mathbf{m}, \F_q)$. Therefore, it is essential to study the equivalence of codes defined over different quotients of $R$, namely $R / R H_\Lambda$ and $R' / R' H_{\Lambda'}$, since they define the same matrix space in order to avoid the risk of having the same code represented on different quotient rings. Therefore, we can prove that LRS codes, ATLRS codes, and TLRS codes of TZ-type, regardless of which skew polynomial quotient they are represented in, form families of codes that do not contain pairwise inequivalent codes.

Let $\theta, \theta'$ be generators of the Galois group $\Gal(\F_{q^m}/\F_q)$, and consider the skew polynomial rings $R := \F_{q^m}[x; \theta]$ and $R' := \F_{q^m}[x; \theta']$. Let $\Lambda, \Lambda' \subseteq \F_q^*$, and consider the skew polynomial quotients $R / R H_\Lambda \simeq \Mat(\bfm,\bfm,\F_q)$ and $R' / R' H_{\Lambda'}\simeq \Mat(\bfm,\bfm,\F_q)$.

\begin{theorem} \label{th:inequivalence}
Assume that one of the following holds:
\begin{itemize}
    \item $3 \leq k \leq tm/2$, or
    \item $tm/2 + 1 \leq k \leq tm - 3$, and both $\Lambda$ and $\Lambda'$ are cyclic subgroups of $\F_q^*$.
\end{itemize}
The following pairs of codes are pairwise inequivalent:
\begin{enumerate}
    \item The LRS code $\C_{k,\theta} \subseteq R / R H_\Lambda$ and, for $\eta \neq 0$, the ATLRS code $\C_{k,\theta'}(\eta,\tau) \subseteq R' / R' H_{\Lambda'}$, whenever $\tau \ne \mathrm{id}$ or $m \nmid k$.
    \item The LRS code $\C_{k,\theta} \subseteq R / R H_\Lambda$ and the TLRS code of TZ-type $\mathcal{D}_{k,\theta'}(\gamma) \subseteq R' / R' H_{\Lambda'}$.
    \item For $\eta \neq 0$, the ATLRS code $\C_{k,\theta}(\eta,\tau) \subseteq R / R H_\Lambda$, with $\tau \neq \theta^{m/2}$ or $k \nmid m$, and the TLRS code of TZ-type $\mathcal{D}_{k,\theta'}(\gamma) \subseteq R' / R' H_{\Lambda'}$.
\end{enumerate}
\end{theorem}

\begin{proof}
Assume, by contradiction, that the LRS code $\C_1 = \C_{k,\theta} \subseteq R / R H_\Lambda$ and the ATLRS code $\C_2 = \C_{k,\theta'}(\eta,\tau) \subseteq R' / R' H_{\Lambda'}$, with $\eta \neq 0$, are equivalent. Then, by \Cref{prop:equivcodeimpliesequivideal} and using \Cref{lm:idealizersATLRS} and \Cref{th:idealizersany}, it follows that
\[
\F_{q^m} \simeq \lid(\C_1) \simeq \lid(\C_2) \simeq \Fix(\F_{q^m}, \tau)\]
and 
\[
\F_{q^m} \simeq \rid(\C_1) \simeq \rid(\C_2) \simeq \Fix(\F_{q^m}, \tau^{-1} \circ \theta^k).
\]
These isomorphisms imply that $\tau = \mathrm{id}$ and $m \mid k$, which contradicts the assumptions in (1). Thus, the codes are not equivalent because of Corollary \ref{cor:codeequiv->idealequiv}.

The proofs of (2) and (3) follow analogously, by applying the same contradiction argument and using the idealisers in \Cref{lm:idealizersATLRS}, \Cref{th:idealizersany}, \Cref{lm:idealizersTZ}, and \Cref{th:idealizersTZany}.
\end{proof}

\begin{remark}
We observe that the nuclear parameters can, in some cases, be also used to distinguish inequivalent codes within the same family of MSRD codes. Indeed, let us consider the family of ATLRS codes, the left idealiser of an ATLRS code $
\C = \C_{k,\theta}(\eta,\tau_h) \subseteq R / R H_\Lambda $
is isomorphic to \( \lid(\C) \cong \F_{p^{(me, h)}} \), where \( \tau_h: y \in \F_{q^m} \mapsto y^{p^h} \in \F_{q^m}\). Therefore, if \( h_1, h_2 \) are distinct divisors of \( tm \), then the codes \( \C_{k,\theta}(\eta,\tau_{h_1}) \) and \( \C_{k,\theta}(\eta,\tau_{h_2}) \) are inequivalent. Henceforth, for a fixed \( k \in \{1, \ldots, tm - 1\} \) and \( \eta \in \F_{q^m}^* \), we can assert that there are at least \( d(tm) \) inequivalent ATLRS codes, where \( d(tm) \) denotes the number of positive divisors of \( tm \).
\end{remark}

\section*{Conclusions}

In this paper, we have introduced and studied several new structural invariants for sum-rank metric codes, including generalised idealisers, the centraliser, the center, and a refined notion of linearity. These tools allowed us to develop a general framework for addressing the code equivalence problem in the sum-rank metric setting.

As a central result, we proved that the main known families of MSRD codes, LRS codes, ATLRS codes, and TLRS codes of TZ-type, are pairwise inequivalent. This proves the strength of the proposed invariants and opens a promising new research direction in the classification of sum-rank metric codes.

We list possible questions/problems that we think are interesting:
\begin{itemize}
    \item We introduced the notion of non-degeneracy for both rank-metric and sum-rank metric codes defined in a matrix setting. This notion captures whether a code is embedded in its minimal ambient space, meaning it cannot be isometrically embedded into a smaller space while preserving the weight distribution. It is well-known that the MacWilliams Extension Theorem holds for rank-metric codes, as proved by the counterexample in \cite[Example 2.9]{barra2015macwilliams}; see also \cite{gorla2024macwilliams}. We note that the code used in this example is degenerate. So, with the aid of the results in Section \ref{sec:nondeg}, it is quite natural to ask whether or not the MacWilliams Extension Theorem holds for non-degenerate (sum-)rank metric codes.
    \item In \Cref{th:inequivalence}, we studied the equivalence problem of the largest known families of MSRD codes within the skew polynomial framework, under the condition that if $k > tm/2$, then the evaluation sets $\Lambda$ and $\Lambda'$ must be cyclic subgroups of $\F_q^*$. This condition was necessary in order to apply some duality arguments. It remains an open question whether the results hold when this cyclicity assumption is withdrawn. This could potentially be addressed through a refined analysis of idealisers for MSRD codes outside the cyclic setting. A deeper study would also require a model for computing the duals of ATLRS and TLRS codes when the evaluation set $\Lambda$ is not a cyclic group.
    \item Our analysis has been mainly focused on left and right idealisers. It would be interesting to find the structure of the generalised idealisers $\mathcal{I}_{\mathbf{s}}(\C)$ for arbitrary $\mathbf{s} \in \mathbb{Z}_2^t$ for the known families of MSRD codes.
    \item By \Cref{prop:sizeideal}, we know that the left and right idealisers of any MSRD code in \( \Mat(\bfm, \bfm, \F_q) \), with \( \bfm = (m, \ldots, m) \), are both fields of size at most \( q^m \). It would be of interest to classify MSRD codes whose left and right idealisers are both isomorphic to \( \F_{q^m} \), in the spirit of the classification of MRD codes carried out in \cite{csajbok2020mrd}. Such a classification may lead to the discovery of new families of MSRD codes.

    \item Other families of MSRD codes have been recently constructed in a vector framework, as in \cite{santonastaso2022subspace,santonastaso2023msrd}. Extending our idealiser-based analysis to these codes could help determine their equivalence or inequivalence with the families investigated in this paper, analogous to what has been done for rank-metric counterpart \cite{zini2021scattered}.
\end{itemize}


\section*{Acknowledgments}
The second author is very grateful for the hospitality of the Dipartimento di Meccanica, Matematica e Management of Politecnico di Bari, he was visiting it during the development of this research in May 2025.
This research was partially supported by the Italian National Group for Algebraic and Geometric Structures and their Applications (GNSAGA - INdAM).

\bibliographystyle{abbrv}
\bibliography{biblio}
\end{document}